\def\typeout#1{}
\documentclass{LMCS}
\usepackage{amsmath}
\usepackage{array}
\usepackage{tikz}
\usetikzlibrary{arrows}
\usetikzlibrary{snakes}
\usepackage{subfigure}
\usepackage{boxedminipage}
\usepackage{enumerate}
\usepackage{defs,hyperref}
\IfFileExists{microtype.sty}{\usepackage{microtype}}{}

\usepackage{algorithm2e}

\def\doi{4 (2:7) 2008}
\lmcsheading%
{\doi}
{1--25}
{}
{}
{Sep.~20, 2007}
{May~15, 2008}
{}   

\begin{document}
\makeatletter
\def\letenv#1#2{\global\newenvironment{#1}{\begin{#2}}{\end{#2}}}
\letenv{theorem}{thm}
\let\c@theorem\c@thm
\letenv{definition}{defi}
\letenv{remark}{rem}
\letenv{lemma}{lem}
\let\c@lemma\c@thm
\letenv{proposition}{prop}
\letenv{corollary}{cor}
\letenv{example}{exa}
\makeatother
\title{On the Expressiveness and Complexity of~\ATL}
\titlecomment{This article is a long version 
  of~\cite{LMO-fossacs07}.}
\author[F.~Laroussinie]{Fran\c cois Laroussinie\rsuper a}
\address{{\lsuper a}LIAFA, Univ.\ Paris 7 \& CNRS, France}
\email{francoisl@liafa.jussieu.fr}
\author[N.~Markey]{Nicolas Markey\rsuper b}
\address{{\lsuper{b,c}}LSV, ENS Cachan \& CNRS, France}
\email{\{markey,oreiby\}@lsv.ens-cachan.fr}
\author[G.~Oreiby]{Ghassan Oreiby\rsuper c}
\thanks{{\lsuper c}This author is sponsored by a PhD grant from Region 
  \^Ile-de-France.}
%
\begin{abstract}
  \ATL is a temporal logic geared towards the specification and
  verification of properties in multi-agents systems. It allows to
  reason on the existence of strategies for coalitions of agents in
  order to enforce a given property.  
  In this paper, we first precisely characterize the complexity of
  \ATL model-checking over Alternating Transition Systems and
  Concurrent Game Structures when the number of agents is not
  fixed. We prove that it is \DD2- and \DD3-complete, depending on the
  underlying multi-agent model (\ATS and \CGS resp.).  We also consider the
  same problems for some extensions of~\ATL.
  We then consider expressiveness issues. We show how \ATS and \CGS are
  related and provide translations between these models w.r.t.\ alternating
  bisimulation. We also prove that the standard definition of~\ATL (built on
  modalities ``Next'', ``Always'' and ``Until'') cannot express the duals of
  its modalities: it is necessary to explicitely add the modality ``Release''.
\end{abstract}

\maketitle

\section{Introduction}

\subsection*{Model checking.}
Temporal logics were proposed for the specification of reactive
systems almost thirty years ago~\cite{CE81,Pnu77,QS82}.  They have
been widely studied and successfully used in many situations,
especially for model checking ---the automatic verification that a
finite-state model of a system satisfies a temporal logic
specification.
Two flavors of temporal logics have mainly been studied:
\emph{linear-time temporal logics}, \emph{e.g.} \LTL~\cite{Pnu77},
which  expresses properties on 
the possible
\emph{executions} of the model; and \emph{branching-time temporal
  logics}, such as \CTL~\cite{CE81,QS82}, which can express
requirements on \emph{states} (which may have several possible
futures) of the model.

\subsection*{Alternating-time temporal logic.}
Over the last ten years, a new flavor of temporal logics has been
defined: \emph{alternating-time temporal
  logics}~(\ATL)~\cite{focs1997-AHK}.
\ATL is a fundamental logic for verifying 
properties in 
\emph{synchronous multi-agent systems}, 
in which several agents can concurrently act upon 
the behavior of the system.
This is particularly interesting for modeling control problems. 
In that setting, it is not only interesting to know if something \emph{can
  arrive} or \emph{will arrive}, as can be expressed in~\CTL or~\LTL, but
rather if some agent(s) can \emph{control} the evolution of the system in order to
enforce a given property. 

The logic \ATL can precisely express this kind of properties, and can
for instance state that ``there is a strategy for a coalition~$A$ of agents
in order to eventually reach an accepting state, whatever the other
agents do''.  \ATL can be seen as an extension of~\CTL; 
its formulae are built on
atomic propositions and boolean combinators, and
(following the seminal 
papers~\cite{focs1997-AHK,compos1997-AHK,jacm49(5)-AHK})
on modalities~$\Diam[A]\X\phi$ (
coalition~$A$ has a
strategy to immediately enter a state satisfying~$\phi$),
$\Diam[A]\G\phi$ (
coalition~$A$ can force  the system 
to always satisfy~$\phi$) and
$\Diam[A]\phi\Until\psi$ (
coalition~$A$ has a strategy
to enforce $\phi\Until\psi$).

\subsection*{Multi-agent models.}
While linear- and branching-time temporal logics are interpreted on Kripke
structure, alternating-time temporal logics are interpreted on models that 
incorporate the notion of \emph{multiple agents}. Two kinds of synchronous
multi-agent models have been proposed for~\ATL in the literature.
First \emph{Alternating Transition
  Systems}~(\ATS{}s)\cite{compos1997-AHK} have been defined: in any
  location of an~\ATS, each agent chooses one \emph{move},
  \emph{i.e.}, a subset of locations (the list of possible moves is
  defined explicitly in the model) in which she would like the
  execution to go~to. When all the agents have made their choice, the
  intersection of their choices is required to contain one single
  location, in which the execution enters.
In the second family of models, called \emph{Concurrent Game
  Structures}~(\CGS{}s)~\cite{jacm49(5)-AHK}, each of the~$n$ agents
has a finite number of possible moves (numbered with integers), and,
in each location, an $n$-ary transition function indicates the state
to which the execution goes.

\subsection*{Our contributions.}
First we precisely characterize the complexity of the model checking
problem. The original works about~\ATL provide model-checking
algorithms in time~$O(m\cdot l)$, where~$m$ is the number of
transitions in the model, and~$l$ is the size of the
formula~\cite{compos1997-AHK,jacm49(5)-AHK}, thus in~\PTIME.
However, contrary to Kripke structures, the number of transitions in a~\CGS or
in an~\ATS is not quadratic in the number of states~\cite{jacm49(5)-AHK}, and
might even be exponential in the number of agents.
\PTIME-completeness thus only holds for~\ATS  when the number of agents is bounded, and
it is shown in~\cite{JD05,JD-ifi} that the problem is strictly\footnote{%
  We adopt the classical hypothesis that the polynomial-time
  hierarchy does not collapse, and that $\PTIME\not=\NP$. 
  We refer to~\cite{Pap94} for the
   definitions about complexity classes, especially about 
  oracle Turing machines and the polynomial-time hierarchy. 
} harder otherwise, namely \NP-hard on~\ATS and \Ss2-hard on~\CGS{}s
where the transition function is encoded as a boolean function.
We prove that it is in fact \DD2-complete and \DD3-complete, resp.
We also precisely characterize the complexity of model-checking classical
extensions of~\ATL, depending on the underlying family of models.

Then we address expressiveness questions.  First we show how \ATS{}s and
\CGS{}s are related by providing  translations between these
models. Moreover we consider expressiveness questions about \ATL
modalities.  While in \LTL and~\CTL, the dual of ``Until'' modality
can be expressed as a disjunction of ``always'' and ``until'', we
prove that it is not the case in~\ATL. In other words, \ATL, as
defined in~\cite{focs1997-AHK,compos1997-AHK,jacm49(5)-AHK}, 
is not as 
expressive as one could expect (while the dual
modalities clearly do not increase the complexity of the verification
problems).

\subsection*{Related works.}
In~\cite{compos1997-AHK,jacm49(5)-AHK}, \ATL has been defined and studied over
\ATS{}s and~\CGS{}s. In~\cite{schobryan}, expressiveness 
issues are considered for $\ATL^*$ and~\ATL. Complexity of
satisfiability is addressed in~\cite{GvD-TCS,WLWW05}.  Complexity
results about model checking (for \ATL, \ATL+, $\ATL{}^*$) can be
found in~\cite{jacm49(5)-AHK,Schobbens-LCMAS03}.
Regarding control- and game theory, many papers have focused on this
wide area;
we refer to~\cite{Wal04} for a survey, and to its numerous references for a
complete overview.

\subsection*{Plan of the paper.} 
Section~\ref{defs} contains the formal definitions needed in
the sequel.  Section~\ref{complex} deals with the model-checking
questions and contains algorithms and complexity analysis for 
\ATS{}s and~\CGS{}s.
Section~\ref{sec-expr} contains our expressiveness results:
we first prove that \ATS{}s and \CGS{}s have
the same expressive power w.r.t.\ alternating bisimulation (\ie, any
\CGS can be translated into an equivalent \ATS, and vice-versa).  
We then present  our expressiveness results concerning 
\ATL modalities.

\section{Definitions}\label{defs}

\subsection{Concurrent Game Structures} 

Concurrent game structures are a multi-player extension of classical Kripke
structures~\cite{jacm49(5)-AHK}. Their definition is as follows:

\begin{definition}\label{def-cgs}
A \emph{Concurrent Game Structure} (\emph{\CGS{}} for short) $\mathcal C$ is a $6$-tuple $(\Agt,
\Loc,\penalty500 \AP,\penalty500 \Lab,\penalty500 \Chc, \Edg)$ where:
\begin{enumerate}[$\bullet$]
\item $\Agt=\{A_1,...,A_k\}$ is a finite set of \emph{agents} (or \emph{players}); 
\item \Loc and~\AP are two finite sets of \emph{locations} and
\emph{atomic propositions}, resp.;
\item $\Lab\colon \Loc\to 2^{\AP}$ is a function labeling each location by the
  set of atomic propositions that hold for that location;
\item $\Chc\colon \Loc\times\Agt \to \Part(\Nat)\smallsetminus\{\varnothing\}$ 
  defines the (finite) set of possible moves of each agent in each location.
\item $\Edg\colon \Loc \times \Nat^k \to \Loc$, where~$k=\size\Agt$, is a
  (partial) function defining the
  transition table. With each location and each set of moves of
  the agents, it associates the resulting location. 
\end{enumerate}
\end{definition}

The intended behaviour is as follows~\cite{jacm49(5)-AHK}: in a 
location~$\ell$, each player~$A_i$ chooses one possible move~$m_{A_i}$
in~$\Chc(\ell, A_i)$ and the next location is given by
$\Edg(\ell,m_{A_1},...,m_{A_k})$. We~write $\Next(\ell)$ for the set
of all possible successor locations from~$\ell$, and
$\Next(\ell,A_j,m)$, with $m\in\Chc(\ell,A_j)$, for the restriction of $\Next(\ell)$ to
locations reachable from~$\ell$ when player~$A_j$ makes the move~$m$.

\medskip
The way the transition table~$\Edg$ is encoded has not been made precise in
the original definition. Following the remarks of~\cite{JD05}, we propose two
possible encodings:

\begin{definition}\label{def-iecgs}\hfill

\begin{enumerate}[$\bullet$]
\item An \emph{\CGSe{}} is a \CGS where the transition table is defined
explicitly. 

\item An \emph{\CGSi{}} is a \CGS where, in each location~$\ell$, 
the transition function is defined by  a
finite sequence $((\phi_0, \ell_0),...,(\phi_n,\ell_n))$, where $\ell_i\in
\Loc$ is a location, and $\phi_i$ is a boolean combination of propositions
$A_j=c$ that 
evaluate to true iff agent~$A_j$ chooses move~$c$. 
The transition table is then defined as follows: 
\( \Edg(\ell,m_{A_1},...,m_{A_k}) = \ell_j \)
iff $j$ is the lowest index s.t.~$\phi_j$ evaluates to true when players~$A_1$
to~$A_k$ choose moves~$m_{A_1}$ to~$m_{A_k}$. We require that 
the last boolean formula~$\phi_n$ be~$\top$, so that no agent can enforce a deadlock.
\end{enumerate}
\end{definition}

Besides the theoretical aspect, the implicit description of~\CGS{}s may
reveal useful in practice, as it allows to not explicitly describe the
full transition table. 

The size $|\mathcal{C}|$ of a \CGS~$\mathcal{C}$ is defined as
$\size{\Loc}+\size{\Edg}$. For \CGSe{}s, $\size{\Edg}$ is the size of the
transition table. For \CGSi{}s, $\size{\Edg}$ is the sum
of the sizes of the formulas used for the definition of~$\Edg$.

\subsection{Alternating Transition Systems}
In the original works about~\ATL~\cite{focs1997-AHK}, 
the logic was interpreted on \ATS{}s, which are
transition systems slightly different from~\CGS{}s:

\begin{definition}\label{def-ats}
An \emph{Alternating Transition System} (\emph{\ATS{}} for short)
$\mathcal A$ is a $5$-tuple $(\Agt, \Loc,\penalty500 \AP,\penalty500
 \Lab, \Chc)$ where:
\begin{enumerate}[$\bullet$]
\item \Agt, \Loc, \AP and \Lab have the same meaning as in~\CGS{}s;
\item $\Chc\colon \Loc\times \Agt \to \Part(\Part(\Loc))$ associate
  with each location~$\ell$ and each agent~$a$ the set of possible
  moves, each move being a subset of~\Loc. For each location~$\ell$,
  it is required that, for any $Q_i\in \Chc(\ell,A_i)$,
  $\bigcap_{i\leq k} Q_i$ be a singleton.
\end{enumerate}
\end{definition}

The intuition is as follows: in a location~$\ell$, once all the agents
have chosen their moves (\ie, a subset of locations), the execution
goes to the (only) state that belongs to all the sets chosen by the
players. Again $\Next(\ell)$ (resp.~$\Next(\ell,A_j,m)$) 
denotes the set of all possible
successor locations (resp.~the set of possible successor locations
when player~$A_j$ 
chooses the move~$m$).

The size of an \ATS is $\size\Loc+\size\Chc$ where $\size\Chc$ is the sum of
the number of locations in each possible move of each agent in each location.

We prove in Section~\ref{traduc-models} that
\CGS{}s and \ATS{}s have the same expressiveness
(w.r.t.~alternating bisimilarity~\cite{AHKV-concur98}).

\subsection{Coalition, strategy, outcomes of a strategy}

A coalition is a subset of agents. In~multi-agent systems, 
a~coalition~$A$ plays against its opponent coalition~$\Agt\smallsetminus A$ 
as if they were two single players.
We thus extend~$\Chc$ and~$\Next$ to coalitions:

\begin{enumerate}[$\bullet$]
\item Given $A\subseteq \Agt$ and $\ell\in\Loc$, $\Chc(\ell,A)$
denotes the possible moves for the coalition $A$ from~$\ell$. Such a
move~$m$ is composed of a single move for every agent of the
coalition, that is $m \egdef (m_a)_{a\in A}$.  Then, given a move
$m'\in\Chc(\ell,\Agt\backslash A)$, we~use $m\oplus m'$ to denote the
corresponding \emph{complete} move (one~for each~agent). In~\ATS{}s, such
a move $m\oplus m'$ corresponds to the unique resulting location;
in~\CGS{}s, it is given by $\Edg(\ell,m\oplus m')$.
\item $\Next$ is extended to coalitions in a natural way: given $m =
(m_a)_{a\in A}\in \Chc(\ell,A)$, we let $\Next(\ell,A,m)$ denote the
restriction of $\Next(\ell)$ to locations reachable from~$\ell$ when
every player $A_j\in A$ makes the move~$m_{A_j}$.
\end{enumerate}

\noindent Let $\calS$ be a \CGS or an \ATS.  A \emph{computation} of~$\calS$ is
an infinite sequence $\rho=\ell_0\ell_1\cdots$ of locations such that
for any $i$, $\ell_{i+1} \in \Next(\ell_i)$. We~write
$\rho[i]$ for the~$i+1$-st location~$\ell_i$.
A~\emph{strategy} for a player $A_i\in \Agt$ is a function~$f_{A_i}$
that maps any finite prefix of a computation to a possible move
for~$A_i$, \ie, satisfying $f_{A_i}(\ell_0\cdots\ell_m) \in
  \Chc(\ell_m,A_i)$.  A~strategy is \emph{state-based} 
(or~\emph{memoryless}) if it only depends on the current state
(\ie, $f_{A_i}(\ell_0\cdots\ell_m) = f_{A_i}(\ell_m)$).

A strategy induces a set of computations from~$\ell$ ---called
the \emph{outcomes} of~$f_{A_i}$ from~$\ell$ and denoted%
\footnote{We might omit to mention~$\calS$ when it is clear from the context.}
\newcounter{fnote}\setcounter{fnote}{\thefootnote}%
$\Out_{\calS}(\ell,f_{A_i})$--- that player~$A_i$ can enforce:
$\ell_0\ell_1\cdots \in \Out_{\calS}(\ell,f_{A_i})$ iff $\ell=\ell_0$ and
for any~$i$ we have
$\ell_{i+1}\in\Next(\ell_i,A_i,f_{A_i}(\ell_0\cdots\ell_i))$.
Given~a coalition~$A\subseteq \Agt$, a strategy for~$A$ is a
tuple~$F_A$ containing one strategy for each player in~$A$: $F_A=\{
f_{A_j} | A_j\in A\}$. The~outcomes of~$F_A$ from a location~$\ell$
contains the computations enforced by the strategies in~$F_A$:
$\ell_0\ell_1\cdots \in \Out_{\calS}(\ell,F_A)$ iff $\ell=\ell_0$ and
for any~$i$, $\ell_{i+1} \in \Next(\ell_i,A,(f_{a} (\ell_0, \cdots
,\ell_i))_{a\in A})$.
The set of strategies for~$A$ is
denoted\footnotemark[\thefnote]~$\Strat_{\calS}(A)$. 
Finally, note that $F_\emptyset$ is empty and
 $\Out_{\calS}(\ell,\emptyset)$ represents the set of all
computations from~$\ell$.

\subsection{The logic \texorpdfstring{\ATL}{ATL}} 

We now define the logic \ATL, whose purpose is to express controllability
properties on \CGS{}s and~\ATS{}s. Our definition is slightly different from
the one proposed in~\cite{jacm49(5)-AHK}. This difference will be explained
and argued in Section~\ref{expr}.
\begin{definition}\label{def-ATL}
The syntax of~\ATL is defined by the following grammar:
\begin{eqnarray*}
\ATL \ni \phi_s,\psi_s& ::=&\top\,\mid\, P \,\mid\, \non\phi_s \,\mid\, \phi_s\ou\psi_s
   \,\mid\, \Diam[A]\phi_p   \\
          \phi_p& ::= &  \non\phi_p \,\mid\, \X\phi_s\,\mid\,  \phi_s\Until\psi_s
\end{eqnarray*}
where $P$ ranges over the set~$\AP$ and $A$ over the subsets of~$\Agt$.

\medskip
Given a formula~$\phi\in \ATL$, the size
of~$\phi$, denoted by~$\size\phi$, is the size of the tree representing that
formula. The DAG-size of~$\phi$ is the size of the directed acyclic graph
representing that formula (\ie, sharing common subformulas).
\end{definition}
In addition, we use standard abbreviations such as $\top$, $\bot$, 
$\F$,~etc.  
\ATL~formulae are interpreted over states of a game structure~$\calS$.  The 
semantics of the main operators is defined 
as follows\footnotemark[\thefnote]:
\begin{xalignat*}1
\ell \sat_{\calS}  \Diam[A] \phi_p & \qquad\mbox{iff}\qquad  
  \exists F_A\in\Strat(A).\ \forall \rho\in\Out(\ell,F_A).\  \rho \sat_{\calS} \phi_p, \\
%
\rho \sat_{\calS}  \X \phi_s & \qquad\mbox{iff}\qquad    \rho[1] \sat_{\calS} \phi_s, \\
\rho \sat_{\calS}  \phi_s \Until \psi_s & \qquad\mbox{iff}\qquad  \exists i.\  
  \rho[i] \sat_{\calS} \psi_s \ \mbox{and}\   \forall 0\leq j < i.\ 
  \rho[j]\sat_{\calS} \phi_s.
\end{xalignat*} 
It is well-known that, for the logic~\ATL, it is sufficient to restrict to
state-based strategies (\emph{i.e.}, $\Diam[A]\phi_p$ is satisfied iff there
is a state-based strategy all of whose outcomes
satisfy~$\phi_p$)~\cite{jacm49(5)-AHK,Schobbens-LCMAS03}.  

Note that $\Diam[\emptyset] \phi_p$ corresponds to the \CTL
formula~$\All \phi_p$ (\ie, universal quantification over all
computations issued from the current state), while $\Diam[\Agt]
\phi_p$ corresponds to existential quantification~$\Ex\phi_p$. 
However, $\non\Diam[A]\phi_p$ is generally \emph{not} equivalent
to~$\Diam[\Agt\smallsetminus
A]\non\phi_p$~\cite{jacm49(5)-AHK,GvD-TCS}: indeed the absence of a
strategy for a coalition $A$ to ensure $\phi$ does not entail the
existence of a strategy for the coalition $\Agt\backslash A$ to ensure
$\non\phi$. For~instance, Fig.~\ref{fig-exCGS} displays a (graphical
representation of a) $2$-player~\CGS for which, in~$\ell_0$, both
$\non\Diam[A_1]\X p$ and $\non\Diam[A_2]\non\X p$ hold. In~such a
representation, a transition is labeled with $\listof{m_1,m_2}$ when
it corresponds to move~$m_1$ of player~$A_1$ and to move~$m_2$ of
player~$A_2$.  Fig.~\ref{fig-exATS} represents an
``equivalent'' 
\ATS with the same property.
\begin{figure}[!ht]
\begin{minipage}[b]{.45\linewidth}
\centering
\begin{tikzpicture}
\path (0,0) node[draw,circle] (A) {$\phantom p$} -- +(60:5mm) node {$\scriptstyle \ell_0$};
\path (-30:2cm) node[draw,circle] (B) {\hbox to 2mm{\hss $p$\hss}} 
  -- +(-30:5mm) node {$\scriptstyle \ell_1$};
\path (-70:2cm) node[draw,circle] (C) {\hbox to 2mm{\hss $\non p$\hss}} 
  -- +(-30:5mm) node {$\scriptstyle \ell'_1$};
\path (-110:2cm) node[draw,circle] (D) {\hbox to 2mm{\hss $\non p$\hss}} 
  -- +(-30:5mm) node {$\scriptstyle \ell'_2$};
\path (-150:2cm) node[draw,circle] (E) {\hbox to 2mm{\hss $p$\hss}} 
  -- +(-30:5mm) node {$\scriptstyle \ell_2$};
\draw[-latex'] (A) -- (B) node[midway,above] {$\scriptstyle \listof{1,1}$};
\draw[-latex'] (A) -- (C) node[pos=.7,right=-1mm] {$\scriptstyle \listof{1,2}$};
\draw[-latex'] (A) -- (D) node[pos=.7,left=-1mm] {$\scriptstyle \listof{2,1}$};
\draw[-latex'] (A) -- (E) node[midway,above] {$\scriptstyle \listof{2,2}$};
\end{tikzpicture}
\caption{A \CGS that is not determined.}\label{fig-exCGS}
\end{minipage}
\hfill
\begin{minipage}[b]{.5\linewidth}
\[\Loc=\{\ell_0,\ell_1,\ell_2,\ell'_1,\ell'_2\}\]
\begin{xalignat*}1
\Chc(\ell_0,A_1) &= \{\{\ell_1,\ell'_1\},\{\ell_2,\ell'_2\}\} \\
\Chc(\ell_0,A_2) &= \{\{\ell_1,\ell'_2\},\{\ell_2,\ell'_1\}\} 
\end{xalignat*}
with $
\begin{cases}
\Lab(\ell_1)=\Lab(\ell_2) = \{p\} \\
\Lab(\ell'_1)=\Lab(\ell'_2) = \varnothing
\end{cases}
$
\caption{An \ATS that is not determined.}\label{fig-exATS}
\end{minipage}
\end{figure}

\section{Complexity of \texorpdfstring{\ATL}{ATL} model-checking}\label{complex}

In this section, we establish the precise complexity of \ATL
model-checking.  This issue has already been addressed in the seminal papers
about~\ATL, on both \ATS{}s~\cite{compos1997-AHK}
and~\CGS{}s~\cite{jacm49(5)-AHK}.  The time complexity is shown to be
in~$O(m\cdot l)$, where $m$ is the number of transitions and $l$ is
the size of the formula.  The authors then claim that the
model-checking problem is in~\PTIME{} (and obviously, \PTIME-complete,
since it is already for~\CTL). In~fact this only holds for explicit
\CGS{}s. In~\ATS{}s, the number of transitions might be exponential in the
size of the system (more~precisely, in the number of~agents).
This problem ---the exponential blow-up
of the number of transitions to handle in the verification
algorithm--- also occurs for implicit~\CGS{}s: the~standard
algorithms running in~$O(m\cdot l)$ require exponential time.

Basically, the algorithm for model-checking~\ATL is similar to that
for~\CTL: it~consists in recursively computing fixpoints, based
\emph{e.g.}~on the following equivalence:
\begin{equation}
\Diam[A] p \Until q \equiv \mu Z. (q \ou (p \et \Diam[A] \X Z))
\label{eq-fixpoint}
\end{equation}
The difference with~\CTL is that we have to deal with the modality
$\Diam[A] \X$ ---corresponding to the \emph{pre-image} of a set of states
\emph{for some coalition}--- instead of the standard modality~$\E\X$. 
In~control theory,  $\Diam[A] \X$ corresponds to
the \emph{controllable predecessors} of a set of states for a
coalition: $\CPre(A,S)$, with $A\subseteq \Agt$ and $S\subseteq \Loc$,
is defined as follows:
\[
\CPre(A,S) \mathrel{\egdef} \{ \ell \in \Loc \mid \exists m_A \in
\Chc(\ell,A) \:\mbox{s.t.}\: \Next(\ell,A,m_A)\subseteq S\}
\]
The crucial point of the model-checking algorithm is the computation
of the set~$\CPre(A,S)$. 

In the sequel, we establish the exact complexity of computing~$\CPre$ (more
precisely, given~$A\subset\Agt$, $S\subseteq\Loc$, and~$\ell\in\Loc$,
the complexity of deciding whether $\ell\in\CPre(A,S)$), and of \ATL 
model-checking for our three kinds of multi-agent systems.

\subsection{Model checking \texorpdfstring{\ATL}{ATL} on \texorpdfstring{\CGSe}{CGSe}{}s.}\label{ATLsurCGSex}

As already mentionned, the precise complexity of \ATL model-checking over
\CGSe{}s was established in~\cite{jacm49(5)-AHK}:
\begin{theorem}
\ATL model-checking over \CGSe{}s is \PTIME-complete.
\end{theorem}

To our knowledge, the precise complexity of computing \CPre in \CGSe{}s has
never been considered. The best upper bound is \PTIME, which is sufficient for
deriving the \PTIME complexity of \ATL model-checking.

In fact, given a location~$\ell$, a set of locations~$S$ and a coalition~$A$,
deciding whether $\ell\in\CPre(A,S)$ has complexity much lower
than~\PTIME:
\begin{proposition}
Computing \CPre in \CGSe{}s is in~$\AC^0$.
\end{proposition}

\begin{proof}
We begin with precisely defining how the input is encoded as a sequence of bits:
\begin{enumerate}[$\bullet$]
\item the first $\size\Agt$ bits define the coalition: the~$i$-th bit is a~$1$
  iff agent~$A_i$ belongs to~$A$;
\item the following $\size\Loc$ bits of input define the set~$S$;
\item for the sake of simplicity, we assume that all the agents have the same
  number of moves in~$\ell$. We~write~$p$ for that number, which we assume is
  at least~$2$. The transition table~$\Edg(\ell)$ is then given as
  a sequence of $p^k$ sets of $\log(\size\Loc)$ bits.
\end{enumerate}

\noindent As a first step, it is rather easy to modify the input in order to have the
following form:
\begin{enumerate}[$\bullet$]
\item first the $k$ bits defining the coalition;
\item then, a sequence of~$p^k$ bits defining whether the resulting state
  belongs to~$S$.
\end{enumerate}
This is achieved by $p^k$ copies of the same $\AC^0$ circuit.

We now have to build a circuit that will ``compute'' whether coalition~$A$ has
a strategy for ending up in~$S$. Since circuits must only depend on the size
of the input, we cannot design a circuit for coalition~$A$. Instead, we build
one circuit for each possible coalition (their number is exponential in the
number of agents, but polynomial in the size of the input, provided
that~$p\geq 2$), and then select the result corresponding to coalition~$A$.

Thus, for each possible coalition~$B$, we build one circuit whose final node
will evaluate to~$1$ iff $\ell\in\CPre(B,S)$. This~is achieved by an unbounded
fan-in circuit of depth~$2$: at~the first level, we put $p^{\size B}$
AND-nodes, representing each of the $p^{\size B}$ possible moves for
coalition~$B$. Each of those nodes is linked to~$p^{k-\size B}$ bits of the
transition table, corresponding to the set of possible $p^{k-\size B}$ moves
of the opponents. At~the second level, an OR-node is linked to all the nodes
at depth~$1$.

Clearly enough, the OR-node at depth~$2$ evaluates to true iff coalition~$B$
has a strategy to reach~$S$. Moreover, there are $\binom kl$ coalitions
of size~$l$, each of which is handled by a circuit of $p^l+1$~nodes. The
resulting circuit thus has $(p+1)^k+2^k$ nodes, which is polynomial in the
size of the input. This circuit is thus an $\AC^0$ circuit. 

It simply remains to return the result corresponding to the coalition~$A$. 
This is easily achieved in~$\AC^0$. 
\end{proof}

\subsection{Model checking \texorpdfstring{\ATL}{ATL} on \texorpdfstring{\CGSi}{CGSi}{}s.}\label{ATLsurCGS}

Assuming that the transitions issued from~$\ell$ are given ---in the
transition table--- by the 
sequence $((\phi_0,\ell_0),(\phi_1,\ell_1),\ldots,(\phi_n,\ell_n))$,
we have: $\ell \in \CPre(A,S)$ iff there exists
$m_A\in\Chc(\ell,A)$, s.t.~there is no
$m_{\bar{A}}\in\Chc(\ell,\Agt\backslash A)$ and $\ell_i\in
\Loc\backslash S$ s.t.\footnote{Given $m= (m_a)_{a\in A}$ for
  $A\subseteq \Agt$, $\phi[m]$ denotes the formula where every
  proposition "$A_j\!=\!c$'' with $A_j\in A$ is replaced by $\top$
   if $m_{A_j}=c$, and by $\bot$  otherwise. If $A=\Agt$,
  $\phi[m]$ is boolean expression.}  $\phi_i[m_A\oplus
m_{\bar{A}}]\equiv \top$ and $\phi_j[m_A\oplus m_{\bar{A}}]\equiv
\bot$ for any $j<i$.
Thus we look for a move $m_A\in \Chc(\ell,A)$
s.t.\ for all~$m_{\bar A}\in\Chc(\ell,\bar A)$, the negation of 
$\OU_{\ell_i\in\Loc\backslash S} (\phi_i[m_A] \et \ET_{j<i} \non
\phi_j[m_A])$ holds. 

This problem corresponds to an instance of the \Ss2-complete problem
\EQSAT[2]:
\begin{prob}{\EQSAT[2]}
\entree{two families of variables~$X=\{x^1, ..., x^n\}$ and~$Y=\{y^1, ..., y^n\}$,
  a boolean formula $\phi$ on the set of variables $X\cup Y$.}
\sortie{True iff 
  \( \exists X.\ \forall Y.\ \phi\).}
\end{prob}%
And indeed, as a direct corollary of~\cite[Lemma~$1$]{JD05}, we have:

\begin{proposition}
Computing $\CPre$ in \CGSi{}s is \Ss2-complete.
\end{proposition}

\begin{proof}
  The membership in~\Ss2 follows directly the above remarks. A~$\Ss2$ 
  procedure is explicitly described in Algorithm~\ref{alg-cpre-cgsi}. 
\begin{algorithm}[ht]
\SetVline
\SetKwComment{Comment}{//}{}
\textbf{Procedure} \texttt{co-strategy}
	($q$, $(\phi_i,\ell_i)_i$, $\left(m_a\right)_{a\in A}$, $S$)\par
\Comment{checks if the opponents have a co-strategy to $(m_a)_{a\in A}$ to avoid $S$}
\Begin{
	\ForEach{$\bar a \in \bar A$}
	{
		$m_{\bar a} \leftarrow guess(q, \bar a)$\;
	}
	$i \leftarrow 0$\;
        \While{$\non\phi_i(m_a,m_{\bar a})$}
	  {$i \leftarrow i+1$\;}

        \If{$\ell_i\notin S$}{\Return{yes}\;}
	\Else{\Return{no}\;}
}
\textbf{Procedure} \texttt{CPre}($A$, $S$) 
\Begin{
	$W\leftarrow \varnothing$\;
	\ForEach{$q \in \calC$}
	{
		\ForEach{$a \in A$}
		{
			$m_a \leftarrow guess(q, a)$\;		
		}
		\If{not \texttt{\upshape co-strategy}($q$, $(\phi_i,\ell_i)_i$, $\left(m_a\right)_{a\in A}$, $S$)}
			{$W\leftarrow W\cup\{q\}$\;}
	}
	\Return{$W$}\;
}
\caption{Computing $\CPre$\/ on \CGSi.}\label{alg-cpre-cgsi}
\end{algorithm}

Concerning hardness in~$\Ss2$, we directly use the construction
of~\cite[Lemma~$1$]{JD05}: from an instance $\exists X.\ \forall Y.\
\phi$ of \EQSAT[2], one consider an implicit \CGS with three
states~$q_1$, $q_\top$ and~$q_\bot$, and $2n$~agents
$A^1$,~...,~$A^n$, $B^1$,~...,~$B^n$, each having two possible choices
in~$q_1$ and only one choice in~$q_\top$ and~$q_\bot$. The transitions
out of~$q_\top$ and~$q_\bot$ are self-loops. The transitions
from~$q_1$ are given by:\quad \( \delta(q_1)= ((\phi[ x^j \leftarrow
(A^j\stackrel?= 1), y^j \leftarrow (B^j\stackrel?=1)], q_\top)
(\top,q_\bot)).  \)

Then clearly,  $q_1$ belongs to $\CPre(\{A^1,...,A^n\},\{q_\top\})$  
iff there exists a valuation for variables in~$X$ s.t.
$\phi$ is true whatever $B$-agents choose for $Y$.
\end{proof}

The complexity of \ATL model checking over \CGSi is higher: the
proof of \Ss2-hardness of $\CPre(A,S)$ can easily be adapted to prove
\PP2-hardness. Indeed consider the dual (thus \PP2-complete) problem
\AQSAT[2], in which, with the same input, the output is the value of
$\forall X.\ \exists Y.\ \phi$.  Then it suffices to consider the same
\CGSi, and the formula $\non\Diam[A^1,...,A^n] \X \non q_\top$.  It
states that there is no strategy for players~$A^1$ to~$A^n$ to
avoid~$q_\top$: whatever their choice, players~$B^1$ to~$B^n$ can
enforce~$\phi$.

This contradicts the claim in~\cite{JD05} that model checking \ATL
would be \Ss2-complete. In fact there is a flaw in their algorithm
about the way it handles negation (and indeed their result holds only for
the \emph{positive} fragment of \ATL~\cite{jamroga-dix-2008}): games
played on~\CGS{}s (and~\ATS{}s) are generally not determined, and the
fact that a player has no strategy to enforce~$\phi$ does not imply
that the other players have a strategy to enforce~$\non\phi$.  It
rather means that the other players have a \emph{co-strategy} to
enforce~$\non\phi$ (by~a \emph{co-strategy}, we mean a way to react to
each move of their opponents~\cite{GvD-TCS}).

\smallskip
Still, using the expression of \ATL modalities as fixpoint formulas 
(see~Eq.~\eqref{eq-fixpoint}), we can compute the set of states
satisfying an \ATL formula by a polynomial number of computations
of~$\CPre$, which yields a $\DD3$ algorithm:

\begin{proposition}\label{implicit-easy}
Model checking \ATL on \CGSi{}s is in~$\DD3$.
\end{proposition}
Note that, since the algorithm consists in labeling the locations with the
subformulae it satisfies, that complexity holds even if we consider
the DAG-size of the formula.

To prove hardness in~$\DD3$, we introduce the following \DD3-complete
problem~\cite{fossacs2001-LMS,Schobbens-LCMAS03}.

\begin{prob}{\SNSAT[2]}
\entree{ $m$ families of variables~$X_i=\{x_i^1, ..., x_i^n\}$,
   $m$ families of variables~$Y_i=\{y_i^1, ..., y_i^n\}$,
   $m$ variables $z_i$,
   $m$ boolean formulae $\phi_i$, with $\phi_i$ involving variables in
   $X_i\cup Y_i \cup\{z_1,...,z_{i-1}\}$.}
\sortie{The value of~$z_m$, defined by \par
\(\left\{\begin{array}{r>{\ }c<{\ }l}
  z_1 &\eqdef& \exists X_1.\ \forall Y_1.\ \phi_1(X_1,Y_1) \\
  z_2 &\eqdef& \exists X_2.\ \forall Y_2.\ \phi_2(z_1,X_2,Y_2) \\
  && \ldots \\
  z_m &\eqdef& \exists X_m.\ \forall Y_m.\ \phi_m(z_1,...,z_{m-1},X_m,Y_m) \\
  \end{array}\right.\)}
\end{prob}%

And we have:

\begin{proposition}\label{implicit-hard}
Model checking \ATL on \CGSi{}s is \DD3-hard.
\end{proposition}

\begin{proof}
  We pick an instance~$\calI$ of  \SNSAT[2], and reduce it to an
  instance of the \ATL model-checking problem.  Note that such an
  instance uniquely defines the values of variables~$z_i$.  We write
  $v_\calI\colon \{z_1,...,z_m\}\to \{\top,\bot\}$ for this valuation.
  Also, when $v_\calI(z_i)=\top$, there exists a witnessing valuation
  for variables in~$X_i$. We extend~$v_\calI$ to $\{z_1,...,z_m\}\cup
  \bigcup_i X_i$, with $v_\calI(x_i^j)$ being a witnessing valuation
  if $v_\calI(z_i)=\top$.

  We now define an \CGSi $\mathcal C$ as follows: it contains $mn$
  agents $A_i^j$ (one for each~$x_i^j$), $mn$ agents $B_i^j$ (one for
  each~$y_i^j$), $m$ agents~$C_i$ (one for each~$z_i$), and one extra
  agent~$D$. The~structure is made of $m$~states~$q_i$,
  $m$~states~$\overline{q_i}$, $m$~states~$s_i$, and two
  states~$q_\top$ and~$q_\bot$.  There are three atomic propositions:
  $s_\top$ and~$s_\bot$, that label the states~$q_\top$ and~$q_\bot$
  resp., and an atomic proposition~$s$ labeling states~$s_i$. The~other states
  carry no label. 

  Except for~$D$, the agents represent booleans, and thus always have
  two possible choices ($0$~and~$1$). Agent~$D$ always has~$m$~choices
  ($0$~to~$m-1$).  The~transition relation is defined as follows: for
  each~$i$,
\[
\begin{array}{c@{\qquad}c}
\begin{array}{l}
\delta(\overline{q_i}) = ((\top, s_i)); \\[1mm]
\delta(s_i) = ((\top, q_i)); \\[1mm]
\delta(q_\top) = ((\top, q_\top)); \\[1mm]
\delta(q_\bot) = ((\top, q_\bot)); 
\end{array}
&
\delta(q_i) = 
\left(
\begin{array}{l}
  ((D\stackrel?=0) \et \phi_i[x_i^j \leftarrow (A_i^j\stackrel?=1), \\
     \qquad y_i^j \leftarrow (B_i^j\stackrel?=1), 
     z_k \leftarrow (C_k\stackrel?=1)],  q_\top) \\
  ((D\stackrel?=0),  q_\bot) \\
  ((D\stackrel?=k) \et (C_k\stackrel?=1),  q_k)  \text{ for each }k<i\\
  ((D\stackrel?=k) \et (C_k\stackrel?=0),  \overline{q_k})  \text{ for each }k<i\\
  (\top, q_\top)
\end{array}
\right)
\end{array}
\]
Intuitively, from state~$q_i$, the boolean agents chose a valuation for the
variable they represent, and agent~$D$ can either choose to check if the
valuation really witnesses~$\phi_i$ (by choosing move~$0$), or
``challenge'' player~$C_k$, with move~$k<i$.

The \ATL formula is built recursively as follows: 
\begin{xalignat*}1
\psi_0 &\eqdef \top \\
\psi_{k+1} &\eqdef \Diam[\Coal] (\non s) \Until (q_\top \ou \Ex\X(s\et\Ex\X \non \psi_k))
\end{xalignat*}
where~$\Coal$ stands for the coalition~$\{A_1^1,..., A_m^n, C_1,...,C_m\}$.

Let $f_{\calI}(A)$ be the state-based strategy for agent~$A\in\Coal$ that 
consists in playing according to the valuation~$v_\calI$ (\emph{i.e.} 
move~$0$ if the variable associated with~$A$ evaluates to~$0$ in~$v_\calI$, and 
move~$1$ otherwise).
The following lemma 
 completes the proof of Proposition~\ref{implicit-hard}:
\newcounter{cptlem}
\setcounter{cptlem}{\value{lemma}}
\begin{lemma}\label{techlemma}
For any $i\leq m$ and~$k\geq i$, the following three statements are equivalent:
\begin{enumerate}[\em(a)]
\item \label{st1}
$\mathcal C, q_i\models\psi_{k}$;
\item \label{st2}
the strategies~$f_\calI$ witness the fact that $\mathcal C, q_i\models\psi_{k}$;
\item \label{st3}
variable $z_i$ evaluates to~$\top$ in~$v_\calI$.
\end{enumerate}
\end{lemma}

%
%
\begin{proof}
Clearly, \iteqref{st2} implies~\iteqref{st1}. 
We prove that \iteqref{st1} implies~\iteqref{st3}
and that \iteqref{st3} implies~\iteqref{st2} 
by induction on~$i$.\medskip

First assume that $q_1\models\psi_{j}$, for some~$j\geq 1$. Since
only $q_\top$ and~$q_\bot$ are reachable from~$q_1$, we have
$q_1\models \Diam[\Coal] \X  q_\top$. 
We are (almost) in the same case as in the \Ss2 reduction of~\cite{JD05}:
there is a valuation of the variables $x_1^1$ to~$x_1^n$ s.t., whatever
players~$D$ and~$B_1^1$ to~$B_m^n$ decide, the run will end up in~$q_\top$. 
This holds in particular if player~$D$ chooses move~$0$: for any valuation of
the variables $y_1^1$ to~$y_1^n$, $\psi_1(X_1,Y_1)$ holds true, and~$z_1$
evaluates to true in~$v_\calI$.

Secondly, if $z_1$ evaluates to true, then $v_\calI(x_1^1)$,~...,
$v_\calI(x_1^n)$ are such that, whatever the value of~$y_1^1$ to~$y_1^n$, $\psi_1$ holds
true. If players~$A_1^1$ to~$A_1^n$ play according to~$f_\calI$, then players~$D$
and~$B_1^1$ to~$B_1^n$ cannot avoid state~$q_\top$, and $q_1\models
\Diam[\Coal] \X  q_\top$, thus also~$\psi_k$
when~$k\geq 1$. 
\smallskip

We now assume the result holds up to index~$i \geq 1$, and prove that it also
holds at step~$i+1$. Assume $q_{i+1}\models\psi_{k+1}$, with $k\geq i$. 
There exists a strategy witnessing~$\psi_{k+1}$, \emph{i.e.},
s.t. all the outcomes following this strategy satisfy
$(\non s) \Until (q_\top \ou \Ex\X(s\et\Ex\X \non \psi_{k}))$. 
Depending on the move of player~$D$ in state~$q_{i+1}$, we get several informations:
 first, if player~$D$ plays move~$l$, with $1\leq l\leq i$, the play
 goes to state~$q_l$ or~$\overline{q_l}$, depending on the choice of
 player~$C_l$. 
  \begin{enumerate}[$\bullet$]
  \item if player~$C_l$ chose move~$0$, the run ends up in~$\overline{q_l}$.
  Since the only way out of that state is to enter state~$s_{l}$, labeled
  by~$s$, we get that $\overline{q_l}\models\Ex\X(s\et\Ex\X \non\psi_{k})$,
  \emph{i.e.}, that $q_l \models \non\psi_{k}$. By~i.h., we get that $z_l$
  evaluates to false in our instance of~\SNSAT[2]. 
  \item if player~$C_l$ chose move~$1$, the run goes to~$q_l$. In that state,
  players in~$\Coal$ can keep on applying their strategy, which ensures that 
  $q_l\models\psi_{k+1}$, and, by~i.h., 
  that $z_l$ evaluates to true in~$\calI$.
  \end{enumerate} 
  Thus, the strategy for~$\Coal$ to enforce~$\psi_{k+1}$ in~$q_{i+1}$
  requires players~$C_1$ to~$C_i$ to play according to~$v_\calI$ and
  the validity of these choices can be verified by the ``opponent''
  $D$.

Now, if player~$D$ chooses move~$0$, all the possible outcomes will necessarily
immediately go to~$q_\top$ (since $\psi_{k+1}$ holds, and
since~$q_\bot\not\models\Ex\X(s\et\Ex\X\non\psi_k)$). We immediately get that
players~$B_{i+1}^1$ to~$B_{i+1}^n$ cannot make~$\psi_{i+1}$ false, hence that
$z_{i+1}$ evaluates to true in~$\calI$.

Secondly, if $z_{i+1}$ evaluates to true, assume players in~$\Coal$ play
according to~$f_\calI$, and consider the possible moves of player~$D$:
\begin{enumerate}[$\bullet$]
\item if player~$D$ chooses move~$0$, since~$z_{i+1}$ evaluates to true and
  since players~$C_1$ to~$C_i$ and $A_{i+1}^1$ to~$A_{i+1}^n$ have played 
  according~$v_\calI$, there is no way for player~$B_{i+1}^1$ to~$B_{i+1}^n$ to avoid
  state~$q_\top$. 
\item if player~$D$ chooses some move~$l$ between~$1$ and~$i$, the execution 
  will go into state~$q_l$ or~$\overline{q_l}$, depending on the move
  of~$C_l$. 
 \begin{enumerate}[$-$]
 \item if $C_l$ played move~$0$, \emph{i.e.}, if $z_l$ evaluates to false
   in~$v_\calI$, the
   execution goes to state~$\overline{q_l}$, and we know by~i.h. that
   $q_l\models\non\psi_{k}$. Thus $\overline{q_l} \models 
   \Ex\X(s\et\Ex\X \non\psi_{k})$, and the strategy still fulfills the
   requirement.
 \item if $C_l$ played move~$1$, \emph{i.e.}, if $z_l$ evaluates to true, then
 the execution ends up in state~$q_l$, in which, by~i.h., the strategy~$f_\calI$ 
enforces~$\psi_{k+1}$. 
 \end{enumerate}
\item if player~$D$ plays some move~$l$ with~$l>i$, the execution goes
  directly to~$q_\top$, and the formula is fulfilled.\qed
\end{enumerate}
\end{proof}
%
%
\end{proof}

With Proposition~\ref{implicit-easy}, this implies:
\begin{theorem}
Model checking \ATL on \CGSi{}s is \DD3-complete.
\end{theorem}

\subsection{Model checking \texorpdfstring{\ATL}{ATL} on \texorpdfstring{\ATS}{ATS}{}s.}\label{ATLsurATS}

For \ATS{}s also, computing~$\CPre$ (and~thus model-checking~\ATL)
cannot be achieved in~\PTIME.  A~direct corollary
of~\cite[Lemma~$4$]{JD05} is: 

\begin{proposition}
Computing~$\CPre$ in \ATS{}s  is \NP-complete.
\end{proposition}

\begin{proof}
Algorithm~\ref{alg-cpre-ats} shows how to compute $\CPre$
in \NP\  in \ATS{}s: it amounts to guessing a move for each player in the
coalition, and to check whether the resulting possible next states are all in~$S$.

\begin{algorithm}[!ht]
\SetVline
\SetKwComment{Comment}{//}{}
\textbf{Procedure} \texttt{CPre}($A$, $S$)
\Begin{
	$W\leftarrow \varnothing$\;
	\ForEach{$q \in \calC$}
	{
		\ForEach{$a \in  A$}
		{
			\Comment{Guess a move for player $a$ from a state $q$}
			$m_a \leftarrow guess(q, a)$\;
		}
		\If{$\bigcap\limits_{a\in A} m_a \subseteq S$}
		{$W\leftarrow W\cup\{q\}$\;}

	}
	\Return{$W$}\;
}
\caption{Computing $\CPre$\/ for \ATS}\label{alg-cpre-ats}
\end{algorithm}

Again, \NP-hardness follows from~\cite[Lemma~$4$]{JD05}. We propose
here a slightly different proof, that will be a first step to the
\DD2-hardness proof below.

The proof is a direct reduction from $3$SAT: let $\calI=(S^1, ...,
S^n)$ be an instance of~$3$SAT over variables~$X = \{x^1,...,
x^m\}$. We assume that $S^j = \alpha^{j,1} s^{j,1} \ou \alpha^{j,2}
s^{j,2} \ou \alpha^{j,3} s^{j,3}$ where $s^{j,k} \in X$ and
$\alpha^{j,k} \in \{0,1\}$ indicates whether variable~$s^{j,k}$ is
taken negatively~($0$) or positively~($1$). We assume without loss of
generality that no clauses contain both one proposition and its
negation.

With such an instance, we associate the following \ATS~$\calA$. 
It contains $8n+1$~states: one
state~$q$, and, for each clause~$S^j$, eight states $q^{j,0}$ to~$q^{j,7}$.
Intuitively, the state~$q^{j,k}$ corresponds to a clause~$B^{j,k} = k_1 s^{j,1} \ou 
k_2 s^{j,2} \ou k_3 s^{j,3}$, where $k_1k_2k_3$ corresponds to the binary
notation for~$k$. There is only one atomic proposition~$\alpha$ in
our~\ATS: a~state~$q^{j,k}$ is labeled with~$\alpha$ iff it does not correspond to 
clause~$S^j$. By~construction, for each~$j$, only one of the states~$q^{j,0}$
to~$q^{j,7}$ is not labeled with~$\alpha$. 

There are~$m+1$ players, where~$m$~is the number of variables that appear
in~$\calI$. 
With each~$x^i$ is associated a player~$A^i$. The~extra player is named~$D$. 
Only the transitions from~$q$ are relevant for this reduction. We~may assume
that the other states only carry a self-loop.
In~$q$, player~$A^i$ decides the value of~$x^i$. She~can
thus choose between two sets of next states, namely the states corresponding
to clauses that are not made true by her choice: 
\begin{xalignat*}2
\{ q^{j,k} \ \mid\ \forall l\leq 3.\ s^{j,l}\not =x^i \ \text{ or }\ \alpha^{i,l}=0\}
  & \qquad\text{ if }x^i=\top \\
\{ q^{j,k} \ \mid\ \forall l\leq 3.\ s^{j,l}\not =x^i \ \text{ or }\ \alpha^{i,l}=1\}
  & \qquad \text{ if }x^i=\bot 
\end{xalignat*}
Last, player~$D$ has~$n$ choices, namely $\{q^{1,0},...,q^{1,7}\}$
to~$\{q^{n,0},...,q^{n,7}\}$. 

We first prove the singleton requirement for \ATS{}s' transitions: the
intersections of the choices of the agents must be a singleton. Once
players~$A^1$ to~$A^m$ have chosen their moves, all the variables
have been assigned a value. Under that valuation, for each~$j\leq n$, exactly
one clause among $B^{j,0}$ to~$B^{j,7}$ evaluates to false (thanks to our
requirement that a literal cannot appear together with its negation in the same
clause). Intersecting with the choice of player~$D$, we end up with one single
state (corresponding to the only clause, among those chosen by~$D$, that evaluates
to false).

Now, let~$\phi=\Diam[A^1,...,A^m] \X\alpha$. That~$q\models\phi$
indicates that players~$A^1$ to~$A^m$ can choose a valuation for~$x^1$
to~$x^m$ s.t.~player~$D$ will not be able to find a clause of the original
instance (\emph{i.e.},~not labeled with~$\alpha$) that evaluates to false
(\emph{i.e.}, that is not made true by any of the choices of the players~$A^1$
to~$A^m$). In that case, the instance is satisfiable.
Conversely, if the instance is satisfiable, it suffices for the players~$A^1$
to~$A^m$ to play according to a satisfying valuation of variables~$x^1$ to~$x^m$. 
Since this valuation makes all the original clauses true, it yields a strategy
that only leads to states labeled with~$\alpha$.
\end{proof}

As in the case of \CGSi{}s, we combine 
the fixpoint expressions of \ATL modalities together with the \NP\ algorithm for
computing~$\CPre$.
This yield a \DD2 algorithm for full~\ATL:
\begin{proposition}\label{ats-easy}
Model checking \ATL over \ATS{}s is in~\DD2.
\end{proposition}

This turns out to be optimal:
\begin{proposition}\label{ats-hard}
Model checking \ATL on \ATS{}s is \DD2-hard.
\end{proposition}

\begin{proof}
The proof is by a reduction of the 
\DD2-complete problem~\SNSAT~\cite{fossacs2001-LMS}:
\begin{prob}{\SNSAT}
\entree{ $p$ families of variables~$X_r=\{x_r^1, ..., x_r^m\}$,
   $p$ variables $z_r$,
   $p$ boolean formulae $\phi_r$ in $3$-CNF, with $\phi_r$ 
      involving variables in 
   $X_r \cup\{z_1,...,z_{r-1}\}$.}
\sortie{The value of~$z_p$, defined by \par
\(\left\{\begin{array}{r>{\ }c<{\ }l}
  z_1 &\eqdef& \exists X_1.\ \phi_1(X_1) \\
  z_2 &\eqdef& \exists X_2.\ \phi_2(z_1,X_2) \\
  z_3 &\eqdef& \exists X_3.\ \phi_3(z_1,,z_2,X_3) \\
  && \ldots \\
  z_p &\eqdef& \exists X_p.\ \phi_p(z_1,...,z_{p-1},X_p)
  \end{array}\right.\)}
\end{prob}

Let $\calI$ be an instance of \SNSAT, where we assume that each~$\phi_r$ is
made of~$n$ clauses $S_r^1$ 
to~$S_r^n$, with $S_r^j=\alpha_r^{j,1} s_r^{j,1} \ou
\alpha_r^{j,2} s_r^{j,2} \ou \alpha_r^{j,3} s_r^{j,3}$.
Again, such an instance uniquely defines a valuation $v_\calI$ for
variables~$z_1$ to~$z_r$, that can be extended to the whole set of variables
by choosing a witnessing valuation for $x_r^1$ to~$x_r^n$ when $z_r$ evaluates
to true.

We now describe the \ATS~$\mathcal A$: it contains $(8n+3)p$ states:
\begin{enumerate}[$\bullet$]
\item $p$ states~$\overline{q_r}$ and $p$ states~$q_r$,
\item $p$ states~$s_r$,
\item and for each formula~$\phi_r$, for each clause $S_r^j$ of~$\phi_r$,
eight states $q_r^{j,0},...,q_r^{j,7}$, as in the previous reduction.
\end{enumerate}

\noindent States~$s_r$ are labelled with the atomic proposition~$s$, and
states~$q_r^{j,k}$ that do not correspond to clause~$S_r^j$ are labeled
with~$\alpha$. 

There is one player~$A_r^j$ for each variable~$x_r^j$, one player~$C_r$
for each~$z_r$, plus one extra player~$D$. As regards transitions, there are
self-loops on each state~$q_r^{j,k}$, single transitions from
each~$\overline{q_r}$ to the corresponding~$s_r$, and from each $s_r$ to the
corresponding~$q_r$. 
From state~$q_r$,
\begin{enumerate}[$\bullet$]
\item player~$A_r^j$ will choose
the value of variable~$x_r^j$, by selecting one of the following two sets of
states:
\begin{xalignat*}{1}
\{ q_r^{g,k} \ \mid\ \forall l\leq 3.\ s_r^{g,l}\not =x_r^j \ \text{ or }\ \alpha_r^{g,l}=0\}
\cup \{q_t, \overline{q_t} \ \mid\ t<r\}
   &\qquad\text{ if }x_r^j=\top  \\
\{ q_r^{g,k} \ \mid\ \forall l\leq 3.\ s_r^{g,l}\not =x_r^j \
  \text{ or }\ \alpha_r^{g,l}=1\}
\cup \{q_t, \overline{q_t} \ \mid\ t<r\}
 &  \qquad \text{ if }x_r^j=\bot
\end{xalignat*}
Both choices also allow to go to one
of the states~$q_t$ or~$\overline{q_t}$. In~$q_r$, players~$A_t^j$ with
$t\not=r$ have one single choice, which is the whole set of states.

\item player~$C_t$ also chooses for the value of the variable it represents.
  As for players~$A_r^j$, this choice will be expressed by choosing between
  two sets of states corresponding to clauses that are not made true. But
  as in the proof of Prop.~\ref{implicit-hard}, players~$C_t$ will also offer
  the possibility to ``verify'' their choice, by going either to state~$q_t$
  or~$\overline{q_t}$. Formally, this yields two sets of states:
\begin{xalignat*}{1}
\{ q_r^{g,k} \,\mid\,\forall l\leq 3.\ s_r^{g,l}\not =z_t \ \text{ or }\ \alpha_r^{g,l}=0\}
\cup \{q_u, \overline{q_u} \,\mid\,u\not=t\}\cup \{q_t\}
& \quad\text{ if }z_t=\top
\\
\{ q_r^{g,k} \,\mid\,\forall l\leq 3.\ s_r^{g,l}\not =z_t \ \text{ or }\
 \alpha_r^{g,l}=1\}
\cup \{q_u, \overline{q_u} \,\mid\,u\not=t\} \cup\{\overline{q_t}\}
& \quad \text{ if }z_t=\bot 
\end{xalignat*}

\item Last, player~$D$ chooses either to challenge a player~$C_t$, with~$t<r$,
  by choosing the set $\{q_t, \overline{q_t}\}$, or to check that a
  clause~$S_r^j$ is fulfilled, by choosing $\{q_r^{j,0},...,q_r^{j,7}\}$.
\end{enumerate}

\noindent Let us first prove that any choices of all the players yields exactly one
state. It is obvious except for states~$q_r$. 
For a state~$q_r$, let us 
first restrict to the choices of all the players~$A_r^j$ and~$C_r$, then: 
\begin{enumerate}[$\bullet$]
\item if we only consider states~$q_r^{1,0}$ to~$q_r^{n,7}$, the same argument as in
the previous proof ensures that precisely on state per clause is chosen,
\item if we consider states~$q_t$ and~$\overline{q_t}$, the choices of
  players~$B_t$ ensure that exactly one state has been chosen in each
  pair~$\{q_t,\overline{q_t}\}$, for each~$t<r$. 
\end{enumerate}
Clearly, the choice of player~$D$ will select exactly one of the remaining 
states.\smallskip

Now, we build the \ATL formula. 
It is a recursive formula (very similar to the
one used in the proof of Prop.~\ref{implicit-hard}),
defined by~$\psi_0=\top$ and (again writing~$\Coal$ for the set of
players~$\{A_1^1,...,A_p^m,C_1,...,C_p\}$): 
\[
\psi_{r+1}  \eqdef  \Diam[\Coal] (\non s) \Until 
  (\alpha \ou 
\Ex\X(s\et\Ex\X \non \psi_r)).
\]
Then, writing~$f_\calI$ for the state-based strategy associated to~$v_\calI$:
\newcounter{lemmaD}
\setcounter{lemmaD}{\value{lemma}}
\begin{lemma}\label{lemmaD}
For any $r\leq p$ and~$t\geq r$, the following statements are equivalent:
\begin{enumerate}[(a)]
\item\label{sta}  $q_r\models\psi_{t}$;
\item \label{stb} the strategies~$f_\calI$ witness the fact that~$q_r\models\psi_{t}$;
\item \label{stc} variable $z_r$ evaluates to true in~$v_\calI$.
\end{enumerate}
\end{lemma}
%
%
\begin{proof}
We prove by induction on~$r$ that \iteqref{sta} implies~\iteqref{stc} and that \iteqref{stc}
implies~\iteqref{stb}, the last implication being obvious.
For~$r=1$, since no~$s$-state is reachable, it
amounts to the previous proof of \NP-hardness.

Assume the result holds up to index~$r$. Then, if $q_{r+1}\models\psi_{t+1}$ 
for some~$t\geq r$, we pick a strategy
for coalition~$\Coal$ witnessing this property. Again, we consider the
different possible choices available to player~$D$: 
\begin{enumerate}[$\bullet$]
\item if player~$D$ chooses to go to one of~$q_u$ and~$\overline{q_u}$,
  with~$u<r+1$: the execution ends up in~$q_u$ if player~$C_u$ chose to
  set~$z_u$ to true. But in that case, 
  formula~$\psi_{t+1}$ still holds in~$q_u$, which yields by~i.h. that $z_u$
  really evaluates to true in~$v_\calI$. Conversely, the execution ends up
  in~$\overline{q_u}$ if player~$C_u$ set~$z_u$ to false. In that case, we get
  that~$q_u\models\non \psi_{t}$, with $t\geq u$, which entails by~i.h. 
  that $z_u$ evaluates to false. 

  This first case entails that player~$C_1$ to~$C_r$ chose the correct value
  for variables~$z_1$ to~$z_r$. 
\item if player~$D$ chooses a set of eight states corresponding to a
  clause~$S_{r+1}^j$, then the strategy of other players ensures that the
  execution will reach a state labeled 
  with~$\alpha$. As in the previous reduction, this indicates that the
  corresponding clause has been made true by the choices of the other players.
\end{enumerate}
Putting all together, this proves that variable~$z_{r+1}$ evaluates to
true.

Now, if variable~$z_{r+1}$ evaluates to true, 
Assume the players in~$\Coal$ play according to valuation~$f_\calI$.
Then 
\begin{enumerate}[$\bullet$]
\item if player~$D$ chooses to go to a set of states that correspond to a
  clause of~$\phi_{r+1}$, he will necessarily end up in a state that is labeled
  with~$\alpha$, since the clause is made true by the valuation we selected.
\item if player~$D$ chooses to go to one of~$q_u$ or~$\overline{q_u}$, for
  some~$u$, then he will challenge player~$B_u$ to prove that his choice was
  correct. By~i.h., and since player~$B_u$ played according to~$f_\calI$,
  formula~$(\non s) \Until(\alpha \ou
  \Ex\X(s\et\Ex\X\non\psi_{t+1}))$ will be satisfied, for any~$t\geq u$.\qed
\end{enumerate}
\end{proof}
\end{proof}

%
%

We end up with the precise complexity of \ATL model-checking on \ATS{}s: 
\begin{theorem}
Model checking \ATL on \ATS{}s is \DD2-complete.
\end{theorem}


\subsection{Beyond \texorpdfstring{\ATL}{ATL}}

As for classical branching-time temporal logics, we can consider
several extensions of \ATL by allowing more possibilities in the way
of combining quantifiers over strategies and temporal modalities. We
define \ATL*~\cite{jacm49(5)-AHK} as follows:

\begin{definition}
   The syntax of \ATL* is defined by the following grammar:
  \begin{eqnarray*}
    \ATL*\,\ni\varphi_s,\psi_s& ::=& \top \mid P \mid\lnot\varphi_s \mid
    \varphi_s\lor\psi_s\mid\Diam[A]\varphi_p  \\
    \varphi_p, \psi_p& ::= &  \varphi_s \mid \non \varphi_p \mid \varphi_p \ou \psi_p \mid  \X\varphi_p \mid
    \varphi_p \Until\psi_p 
  \end{eqnarray*}
  where $P$ and~$A$ range over $\AP$ and $2^{\Agt}$, resp. 
\end{definition}
The size and DAG-size of an \ATL* formula are defined in the same 
way as for~\ATL. 
\ATL* formulae are interpreted over states of a game
structure~$\calS$, the semantics of the main modalities is as follows 
(if~$\rho=\ell_0\,\ell_1\,\ldots$, we write~$\rho^i$ for the~$i+1$-st 
suffix, starting from~$\ell_i$):
\begin{xalignat*}1 \ell \sat_{\calS} \Diam[A] \phi_p &
  \qquad\mbox{iff}\qquad
  \exists F_A\in\Strat(A).\ \forall \rho\in\Out(\ell,F_A).\  \rho \sat_{\calS} \phi_p, \\
  \rho  \sat_{\calS} \varphi_s  & \qquad\mbox{iff}\qquad \rho[0] \sat_{\calS} \varphi_s  \\
  \rho \sat_{\calS}  \X \phi_p & \qquad\mbox{iff}\qquad    \rho^1 \sat_{\calS} \phi_p, \\
  \rho \sat_{\calS} \varphi_p \Until\psi_p & \qquad\mbox{iff}\qquad
  \rho^i \sat_{\calS} \psi_p \ \mbox{and}\ \forall 0\leq j < i.\
  \rho^j\sat_{\calS} \phi_p 
\end{xalignat*}

\ATL is the fragment of \ATL* where each modality $\Until$ or $\X$
has to be preceded by a strategy quantifier~$\Diam[A]$. 
Several other fragments of \ATL* are also classically defined:
\begin{enumerate}[$\bullet$]
\item \ATL+ is the restriction of \ATL* where a strategy quantifier
  $\Diam[A]$ has to be inserted between two embedded temporal
  modalities $\Until$ or $\X$ but boolean combination are allowed.
\item \EATL extends \ATL by allowing the operators $\Diam[A] \G \F$
  (often denoted as~$\Diam[A]\Finf$) and $\Diam[A] \F \G$ (often written
  $\Diam[A]\Ginf$). They are especially useful to express fairness
  properties.
\end{enumerate}
For instance, 
\begin{xalignat*}1
\Diam[A] \Bigl(\F P_1 \:\et\: \F P_2 \:\et\: P_3 \Until P_4 \Bigr)&\text{ is
  in~\ATL+,} \\
\Diam[A] \F \Bigl(P_1 \:\et \: \Diam[A'] \Finf P_2\Bigr)&\text{ is in \EATL,}\\
\Diam[A] \Bigl(\F P_1 \:\et\: P_2 \Until (P_3\:\et\: \F P_4)\Bigr)&\text{ is in \ATL*.}
\end{xalignat*}

\subsubsection{Model checking \ATL+}
 
First note that \ATL+ extends \ATL and allows to express properties
with more succinct formulae~\cite{wilke99,AI01}  
but these two logics have the same
expressive power: every \ATL+ formula can be translated into an
equivalent \ATL formula~\cite{schobryan}.

The complexity of model checking \ATL+ over \ATS{}s has been settled
\DD3-complete in~\cite{Schobbens-LCMAS03}. But the \DD3-hardness proof
of~\cite{Schobbens-LCMAS03} is in \LOGSPACE{} only 
w.r.t.~the DAG-size of the formula. 
Below, we prove that model checking \ATL+ is \DD3-complete 
(with the classical definition of the size of a formula) 
for our three kinds of game structures. 
%
%
\begin{proposition}
Model checking \ATL+ can be achieved in~\DD3 on \CGSi{}s.
\end{proposition}

\begin{proof}
A \DD3 algorithm is given in~\cite{Schobbens-LCMAS03} for \CGSe{}s.
We extend it to handle \CGSi{}s: for each subformula of
the form~$\Diam[A]\phi$, guess (state-based) strategies for players in~$A$. 
In each state, the choices of each player in~$A$ can be replaced in the
transition functions. We then want to compute the set of states where the
\CTL+ formula $\All\phi$ holds. This can be achieved in \DD2~\cite{CES86,fossacs2001-LMS}, but 
 requires to first compute the possible
transitions in the remaining structure, \emph{i.e.}, to check which of the
transition formulae are satisfiable. This is done by a polynomial number
of independent calls to an \NP{} oracle, and thus does not increase the
complexity of the algorithm.
\end{proof}

\begin{proposition}
Model checking \ATL+ on turn-based two-player \CGSe{}s is \DD3-hard.
\end{proposition}

\begin{proof}
This reduction is a quite
straightforward extension of the one presented in~\cite{fossacs2001-LMS}
for~\CTL+. In particular, it is quite different from the previous reductions,
since the boolean formulae are now encoded in the \ATL+ formula, and not in the
model. 

We  encode an instance~$\calI$ of \SNSAT[2], 
keeping the notations used in the proofs of Prop.~\ref{implicit-hard} (for the
\SNSAT[2] problem) and~\ref{ats-hard} (for clause numbering).
Fig.~\ref{fig-dd3} depicts the turn-based two-player \CGS~$\mathcal C$
associated to~$\calI$. 
\begin{figure*}[!ht]
\centering
\begin{tikzpicture}
  \everymath{\scriptstyle}
  \tikzstyle{every node}=[rectangle, rounded corners]

  \foreach \x / \y / \nod / \lab in 
    { 0.5 / 4.0 /znb/ \overline{z_p} , 
      0.5 / 2.0 /zn/ \vphantom{\overline z}z_p ,
      1.8 / 4.0 /znmob/ \overline{z_{p-1}} , 
      1.8 / 2.0 /znmo/ \vphantom{\overline z}z_{p-1} ,
      3.5 / 4.0 /zub/ \overline{z_{1\vphantom p}} , 
      3.5 / 2.0 /zu/ \vphantom{\overline{z_p}}z_1 }
  { \path (\x,\y)  node[draw] (\nod) {\hbox to 15pt{\hss$\lab$\hss}}; }

  \foreach \x / \y / \nod / \lab in 
    { 0.5 / 3.0 /sn/ {s_p} , 
     1.8 / 3.0 /snmo/ {s_{p-1}} , 
     3.5 / 3.0 /su/ {s_1} }
  {\path (\x,\y)  node[draw] (\nod) {\hbox to 12pt{\hss$\lab$\hss}}; }

  \foreach \x / \y / \nod / \lab in 
    { 4.5 / 3.5 /anb/ \overline{x_1^1} , 
      4.5 / 2.5 /an/ \vphantom{\overline{x^1_1}}x_1^1 ,
      5.5 / 3.5 /anmob/ \overline{x_1^2} , 
      5.5 / 2.5 /anmo/ \vphantom{\overline{x_1^2}}x_1^2 ,
      7.0 / 3.5 /aub/ \overline{x_{m\vphantom{1}}^{n\vphantom{2}}} , 
      7.0 / 2.5 /au/ \vphantom{\overline{x_{m\vphantom{1}}^{n\vphantom{2}}}}x_m^n }
  { \path (\x,\y)  node[draw] (\nod) {\hbox to 12pt{\hss$\lab$\hss}}; }

  \foreach \x / \y / \nod / \lab in 
    { 8.0 / 3.5 /bnb/ \overline{y_1^1} , 
      8.0 / 2.5 /bn/ \vphantom{\overline{y_1^1}}y_1^1 ,
      9.0 / 3.5 /bnmob/ \overline{y_1^2} , 
      9.0 / 2.5 /bnmo/ \vphantom{\overline{y_1^2}}y_1^2 ,
      10.5 / 3.5 /bub/ \overline{y_{m\vphantom{1}}^{n\vphantom{2}}} , 
      10.5 / 2.5 /bu/ \vphantom{\overline{y_{m\vphantom{1}}^{n\vphantom{2}}}}y_m^n }
  { \path (\x,\y)  node[draw] (\nod) {\hbox to 12pt{\hss$\lab$\hss}}; }

  \foreach \na / \nb in
    {znb/znmob,znb/znmo,znb/sn,
     zn/znmob,zn/znmo,
     znmob/snmo,snmo/znmo,sn/zn,
     zub/su,su/zu,zub/anb,zub/an,zu/anb,zu/an,
     anb/anmob,anb/anmo,an/anmob,an/anmo,
     aub/bnb,aub/bn,au/bnb,au/bn,
     bnb/bnmo,bnb/bnmob,bn/bnmo,bn/bnmob}
  { \draw[arrows=-latex'] (\na) -- (\nb); }

  \foreach \nod / \angle in
    {bub/0,bu/0}
  { \draw[arrows=-latex'] (\nod.\angle-20) 
        .. controls +(\angle-20:6mm) and +(\angle+20:6mm) 
        .. (\nod.\angle+20); }

  \foreach \nod / \x / \y in
    {znmob/.5/-.75,znmob/.75/0,znmo/.5/.75,znmo/.75/0,
     anmob/.75/0,anmob/.6/-.4,anmo/.75/0,anmo/.6/.4,
     bnmob/.75/0,bnmob/.6/-.4,bnmo/.75/0,bnmo/.6/.4}
  { \draw[dashed] (\nod) -- +(\x,\y); }

  \foreach \nod / \x / \y in
    {zub/.5/.75,zub/.75/0,zu/.5/-.75,zu/.75/0,
     aub/.75/0,aub/.6/.4,au/.75/0,au/.6/-.4,
     bub/.75/0,bub/.6/.4,bu/.75/0,bu/.6/-.4}
  { \draw[dashed,arrows=latex'-] (\nod) -- +(-\x,-\y); }

  \draw[snake=brace,segment amplitude=2mm] (7.3,1.6) -- (.2,1.6) node[midway,below=2mm]
  {controlled by player~A} ; 
  \draw[snake=brace,segment amplitude=2mm] (10.8,1.6) -- (7.7,1.6) node[midway,below=2mm]
  {controlled by player~B} ;

\end{tikzpicture}
\caption{The CGS $\mathcal C$}
\label{fig-dd3}
\end{figure*}
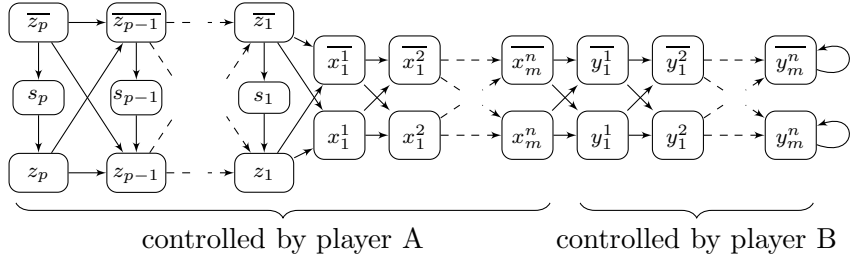
States~$s_1$ to~$s_m$ are labeled by atomic proposition~$s$,
states~$\overline{z_1}$ to~$\overline{z_m}$ are labeled by atomic
proposition~$\overline z$, and the other states are labeled by their name as
shown on Fig.~\ref{fig-dd3}.

The \ATL+ formula is built recursively, with~$\psi_0=\top$ and
\[
\psi_{k+1} = \Diam[A] [\G \non s \et \G(\overline{z} \thn \Ex\X(s\et \Ex\X
  \non\psi_{k})) \et  
 \bigwedge_{w\leq p}[ (\F z_w) \thn \bigwedge_{j\leq n\vphantom k} \bigvee_{k\leq 3} \F
  l_w^{j,k} ] ]
\]
where 
$l_w^{j,k} = v$ when $s_w^{j,k}=v$ and $\alpha_w^{j,k}=1$, and
$l_w^{j,k} = \overline{v}$ when $s_w^{j,k}=v$ and $\alpha_w^{j,k}=0$.
We then have:
\begin{lemma}
For any $r\leq p$ and~$t\geq r$, the following statements are equivalent:
\begin{enumerate}[(a)]
\item\label{staa} $z_r\models\psi_{t}$;
\item\label{stbb} the strategies~$f_\calI$ witness the fact that~$q_r\models\psi_{t}$;
\item\label{stcc} variable $z_r$ evaluates to true in~$v_\calI$.
\end{enumerate}
\end{lemma}

When~$r=1$, since no~$s$- or~$\overline z$-state is
reachable from~$z_1$, the fact that $z_1\models\psi_t$, with~$t\geq 1$, is equivalent
to $z_1\models \Diam[A] \bigwedge_j \bigvee_k \F l_1^{j,k}$. This in turn is
equivalent to the fact that~$z_1$ evaluates to true in~$\calI$.

We now turn to the inductive case. If~$z_{r+1}\models\psi_{t+1}$ with~$t\geq r$,
consider a strategy for~$A$ s.t.~all the outcomes satisfy the property, and
pick one of those outcomes, say~$\rho$. Since it cannot run into any~$s$-state, it
defines a valuation~$v_\rho$ for variables~$z_1$ to~$z_{r+1}$ and~$x_1^1$
to~$x_m^n$ in the obvious way. Each time the outcome runs in
some~$\overline{z_u}$-state, it satisfies~$\Ex\X(s\et\Ex\X\psi_{t})$.
Each time it runs in some~$z_u$-state, the suffix of the outcome
witnesses formula~$\psi_{t+1}$ in~$z_u$. Both cases entail, thanks to the~i.h.,
that $v_\rho(z_u)=v_\calI(z_u)$ for any~$u<r+1$. Now, the subformula
$\bigwedge_{w}[ (\F z_w) \thn \bigwedge_{j\leq n\ } \bigvee_{\ k\leq 3} \F
  l_w^{j,k}$, when~$w=r+1$, 
entails that $\phi_{r+1}$ is indeed satisfied whatever the values of
the~$y_{r+1}^j$'s, \emph{i.e.}, that~$z_{r+1}$ evaluates to true in~$\calI$.

Conversely, if~$z_r$ evaluates to true, then strategy~$f_\calI$ clearly 
witnesses the fact that $\psi_t$ holds in state~$z_r$. 
\end{proof}

As an immediate corollary, we end up with:
\begin{theorem}
Model checking \ATL+ is \DD3-complete on~\ATS{}s as well as on \CGSe{}s and 
\CGSi{}s.
\end{theorem}

%
%

\subsubsection{Model checking \EATL}

In the classical branching-time temporal logics, adding the modality
$\Ex\Finf$ to \CTL increases its expressive power
(see~\cite{emerson90}), this is also true when considering 
alternating-time temporal logics, 
as we will see in Section~\ref{sec-expre-finf}.

From the theoretical-complexity point of view, there is no
difference between \ATL and \EATL:

\begin{theorem}
 Model checking \EATL is:
\begin{enumerate}[$\bullet$]
\item \PTIME-complete over \CGSe{}s;
\item   \DD2-complete over \ATS{}s;
\item   \DD3-complete over \CGSi{}s.
\end{enumerate}
\end{theorem}

\begin{proof}
We extend the model-checking algorithm for \ATL. This is again 
achieved by expressing  modalities~$\Diam[A]\Finf$ 
and~$\Diam[A]\Ginf$ as fixpoint formulas~\cite{dealfaro2001vcd}:
\begin{xalignat*}1
\Diam[A]\Finf p &\equiv \nu y.\mu x. \left(\Diam[A]\X(x)\lor \left( p \land
  \Diam[A]\X(y)\right)\right)
\\
\Diam[A]\Ginf p &\equiv \mu y.\nu x. \left(\Diam[A]\X(x)\land \left( p \lor 
  \Diam[A]\X(y)\right)\right)
\end{xalignat*}

Computing these fixpoints can again be achieved by a polynomial number of 
computations of~$\CPre$.

Hardness directly follows from the hardness of~\ATL model checking.
\end{proof}

\subsubsection{\ATL* model-checking.}

When considering \ATL* model checking, the complexity is the same for
\CGSe, \CGSi and \ATS since it mainly comes from the formula to be
checked:

\begin{theorem}
  Model checking \ATL* is \TWOEXPTIME-complete  on~\ATS{}s as well as on
  \CGSe{}s and \CGSi{}s.
\end{theorem}

\begin{proof}
  We extend the algorithm of~\cite{jacm49(5)-AHK}. This algorithm recursively
  labels each location with the subformulae it satisfies.
  Formulas~$\Diam[A]\psi$, with $\psi\in\LTL$, are handled by building a
  deterministic Rabin tree automaton~$\mathcal A_\psi$ for~$\psi$, and a
  B\"uchi tree automaton~$\mathcal A_{\mathcal C,A}$ recognizing trees
  corresponding to the sets of outcomes of each possible strategy of
  coalition~$A$ in the structure~$\mathcal C$. We~refer
  to~\cite{jacm49(5)-AHK} for more details on the whole proof, and only focus
  on the construction of~$\mathcal A_{\mathcal C,A}$.

  The states of~$\mathcal A_{\mathcal C,A}$ are the states of~$\mathcal C$.
  From location~$\ell$, there are as many transitions as the number of
  possible joint moves~$m=(m_{A_i})_{A_i\in A}$ of coalition~$A$. Each
  transition is a set of states that should appear at the next level of the
  tree. Formally, given~$p\in 2^{\AP}$, 
\[
\delta(\ell,p) = \{\Next(\ell,A,m) \mid m=(m_{A_i})_{A_i\in A} \text{ with }
  \forall A_i\in A.\ m_{A_i} \in \Chc(\ell,A_i) \}
\]
  when~$p=\Lab(\ell)$, and $\delta(\ell,p)=\varnothing$ otherwise.
  
  For \CGSe{}s, this transition function is easily computed in polynomial
  time. For~\ATS{}s and \CGSi{}s, the transition function is computed by
  enumerating the (exponential) set of joint moves of coalition~$A$ (computing
  $\Next(\ell,A,m)$ is polynomial once the joint move is fixed).

  Computing $\mathcal A_{\mathcal C,A}$ can thus be achieved in exponential
  time. Testing the emptiness of the product automaton then requires
  doubly-exponential time. The~whole algorithm thus runs in \TWOEXPTIME.
  The~lower bound directly follows from the lower bound for \CGSe{}s.
\end{proof}

Let us finally mention that 
our results could easily be lifted to
Alternating-time $\mu$-calculus~(AMC)~\cite{jacm49(5)-AHK}: 
the \PTIME algorithm proposed in~\cite{jacm49(5)-AHK} for
\CGSe{}s, which again consists in a polynomial number of computations of
\CPre, is readily adapted to \ATS{}s and \CGSi{}s: as~a result,
model checking the alternation-free fragment has the same complexities as
model checking \ATL, and model checking the whole~AMC is in \EXPTIME for
our three kinds of models.

\section{Expressiveness}\label{sec-expr}

We have seen that the ability of quantifying over the possible
strategies of the agents increases the complexity of model checking
and makes the analysis more difficult. 

We now turn to expressivity issues. We first focus on translations between
our different models (\CGSe, \CGSi and \ATS). We then
consider the expressiveness of ``Until'' and ``Always'' modalities,
proving that they cannot express the dual of ``Until''.

\subsection{Comparing the expressiveness of \CGS{}s and \ATS{}s}
\label{traduc-models}\label{sec:translation}

We prove in this section that \CGS{}s and \ATS{}s are closely related: they
can model the same concurrent games. 
In order to make this statement formal, we use the following definition:
\begin{definition}[\cite{AHKV-concur98}]
Let $\calA$ and~$\calB$ be two models of concurrent games (either \ATS{}s
or~\CGS{}s) over the same set~\Agt of agents. Let $R\subseteq
\Loc_{\calA}\times \Loc_{\calB}$ be a (non-empty) relation between states of~$\calA$ and
states of~$\calB$. That relation is an \emph{alternating bisimulation} when,
for any~$(\ell,\ell')\in R$, the following conditions hold:
\begin{enumerate}[$\bullet$]
\item $\Lab_{\calA}(\ell) = \Lab_{\calB}(\ell')$;
\item for any coalition~$A\subseteq\Agt$, we have
\begin{multline*}
\forall m\colon A \to \Chc_{\calA}(\ell,A).\ \exists m'\colon A \to
\Chc_{\calB}(\ell',A).\\
\forall q'\in\Next(\ell',A,m').\ \exists q\in\Next(\ell,A,m).\ 
  (q,q')\in R.
\end{multline*}
\item symmetrically, for any coalition~$A\subseteq\Agt$, we have
\begin{multline*}
\forall m'\colon A \to \Chc_{\calB}(\ell',A).\ \exists m\colon A \to
\Chc_{\calA}(\ell,A).\\
\forall q\in\Next(\ell,A,m).\ \exists q'\in\Next(\ell',A,m').\ 
  (q,q')\in R.
\end{multline*}
\end{enumerate}
where $\Next(\ell,A,m)$ is the set of locations that are reachable from~$\ell$
when each player~$A_i\in A$ plays~$m(A_i)$. 

Two models are said to be alternating-bisimilar if there exists an alternating
bisimulation involving all of their locations.
\end{definition}

With this equivalence in mind, \ATS{}s and \CGS{}s (both implicit and explicit
ones) have the same expressive power\footnote{The translations between \ATS{}s 
and \CGSe{}s was already mentionned in~\cite{GJ-synth139(2)}.}: 
\newcounter{thmcgsats}
\setcounter{thmcgsats}{\value{theorem}}
\begin{theorem}
\begin{enumerate}
\setcounter{enumiv}{0}
\item\label{eiCGS} Any~explicit \CGS can be translated into an alternating-bisimilar implicit one
in linear time;
\item\label{ieCGS} Any~implicit \CGS can be translated into an alternating-bisimilar explicit one
in exponential time;
\item\label{eATS} Any~explicit \CGS can be translated into an alternating-bisimilar \ATS
in cubic time;
\item\label{ATSe} Any~\ATS can be translated into an alternating-bisimilar explicit \CGS
in exponential time;
\item\label{iATS} Any~implicit \CGS can be translated into an alternating-bisimilar \ATS
in exponential time;
\item\label{ATSi} Any~\ATS can be translated into an alternating-bisimilar implicit \CGS
in quadratic time;
\end{enumerate}
\end{theorem}
Figure~\ref{fig-transl} summarizes those results. From our complexity results
(and the assumption that the polynomial-time hierarchy does not collapse), the
costs of the above translations is optimal. 
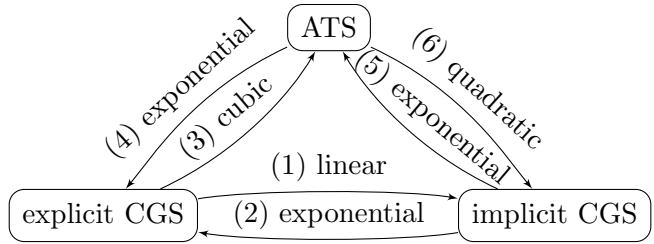
\begin{figure}[!ht]
\centering
\begin{tikzpicture}
  \path (4,2.5) node[draw,rectangle,rounded corners=2mm,inner sep=5pt] (s0) {\ATS};
  \path (1,0) node[draw,rectangle,rounded corners=2mm,inner sep=5pt] (s1) {explicit \CGS};
  \path (7,0) node[draw,rectangle,rounded corners=2mm,inner sep=5pt] (s2) {implicit \CGS};

  \draw[arrows=-latex'] (s0) .. controls +(-155:2cm) and +(50:1cm) .. (s1)
     node[pos=.4, above, sloped] {\eqref{ATSe} exponential};
  \draw[arrows=-latex'] (s1) .. controls +(25:2cm) and +(-120:1cm) .. (s0)
     node[pos=.4, above, sloped] {\eqref{eATS} cubic};
  \draw[arrows=-latex'] (s0) .. controls +(-25:2cm) and +(130:1cm) .. (s2)
     node[pos=.4, above, sloped] {\eqref{ATSi} quadratic};
  \draw[arrows=-latex'] (s2) .. controls +(155:2cm) and +(-60:1cm) .. (s0)
     node[pos=.4, above, sloped] {\eqref{iATS} exponential};
  \draw[arrows=-latex'] (s1) .. controls +(10:2cm) and +(170:2cm) .. (s2)
     node[pos=.5, above, sloped] {\eqref{eiCGS} linear};
  \draw[arrows=-latex'] (s2) .. controls +(-170:2cm) and +(-10:2cm) .. (s1)
     node[midway, above, sloped] {\eqref{ieCGS} exponential};
\end{tikzpicture}
\caption{Costs of translations between the three models}
\label{fig-transl}
\end{figure}

%
%

\begin{proof}
Points~\ref{eiCGS}, \ref{ieCGS}, and~\ref{ATSe} are reasonnably easy. 

For point~\ref{ATSi}, it suffices to write, for each possible next location,
the conjunction (on each agent) of the disjunction of the choices that contain
that next location. 
For instance, if we have $\Chc_{\calA}(\ell_0, A_1) =
\{\{\ell_1,\ell_2\},\{\ell_1,\ell_3\}\}$ and $\Chc_{\calA}(\ell_0, A_2) =
\{\{\ell_2,\ell_3\},\{\ell_1\}\}$ in the \ATS~$\calA$, 
then each player will have two choices in the associated \CGS~$\calB$, and 
\[
\Edg_{\calB}(\ell_0) = 
\left(
\begin{array}{rcl}
(A_1=1\ou A_1=2)\et (A_2=2), &\ell_1 \\
(A_1=1) \et (A_2=1), &\ell_2 \\
(A_1=2) \et (A_2=1), &\ell_3
\end{array}
\right)
\]

Formally, let $\calA =
(\Agt,\Loc_{\calA},\AP,\Lab_{\calA},\Chc_{\calA})$ be an \ATS. We then
define~$\calB=(\Agt,\Loc_{\calB},\AP,\Lab_{\calB},\Chc_{\calB},\Edg_{\calB})$
as follows:
\begin{enumerate}[$\bullet$]
\item $\Loc_{\calB}=\Loc_{\calA}$, $\Lab_{\calB}=\Lab_{\calA}$;
\item $\Chc_{\calB}\colon \ell\times A_i \to
  [1,\size{\Chc_{\calA}(\ell,A_i)}]$;
\item $\Edg_{\calB}$ is a function mapping each location~$\ell$ to 
  the sequence
  $((\phi_{\ell'},\ell'))_{\ell'\in\Loc_{\calA}}$ (the order is not important
  here, as the formulas will be mutually exclusive)  with 
\[
\phi_{\ell'} = \ET_{A_i\in\Agt}\Biggl(
    \OU_{\genfrac{}{}{0pt}{1}{\ell'\text{ appears in the $j$-th}}
    {\text{set of }\Chc_{\calA}(\ell,A_i)}}
     A_i\iseq j
  \Biggr)
\]
\end{enumerate}
Computing $\Edg_{\calB}$ requires quadratic time (more
precisely~$O(\size{\Loc_{\calA}}\times\size{\Chc_{\calA}})$). 
It is now easy to prove that the identity $\Id\subseteq
\Loc_{\calA}\times\Loc_{\calB}$ is an alternating bisimulation, since there is
a direct correspondance between the choices in both structures. 

\medskip
We now explain how to transform an explicit \CGS into an \ATS, showing
point~\ref{eATS}. 
Let $\calA=(\Agt,\Loc_{\calA},\AP,\Lab_{\calA},\Chc_{\calA},\Edg_{\calA})$ be
an explicit \CGS. We define the \ATS $\calB =
(\Agt,\Loc_{\calB},\AP,\Lab_{\calB},\Chc_{\calB})$ as follows (see
Figure~\ref{fig-cgs2ats} for more intuition on the construction):
\begin{enumerate}[$\bullet$]
\item $\Loc_{\calB} \subseteq \Loc_{\calA} \times \Loc_{\calA}\times \Nat^k$,
  where $k=\size\Agt$, with $(\ell,\ell',m_{A_1},\ldots,m_{A_k}) \in
  \Loc_{\calB}$ iff $\ell = \Edg_{\calA}(\ell',m_{A_1},\ldots,m_{A_k})$;
\item $\Lab_{\calB}(\ell,\ell',m_{A_1},\ldots,m_{A_k}) = \Lab_{\calA}(\ell)$;
\item From a location~$q=(\ell,\ell',m_{A_1},\ldots,m_{A_k})$, player~$A_j$ has
  $\size{\Chc_{\calA}(\ell,A_j)}$ possible moves:
\begin{multline*}
\Chc_{\calB}(q,A_j) = \Bigl\{
  \bigl\{ 
    (\ell'',\ell,m'_{A_1},\ldots,m'_{A_{j}}=i,\ldots,m'_{A_k}) 
     \ \mid\ 
  m'_{A_n} \in\Chc_{\calA}(\ell,A_n) \\
  \text{ and }\ell''=\Edg_{\calA}(\ell,m_{A_1},\ldots,m_{A_{j}}=i,\ldots,m_{A_k})
  \bigr\} \ \mid\ i\in\Chc_{\calA}(\ell,A_j)
\Bigr\}
\end{multline*}
\end{enumerate}
This \ATS is built in time $O(\size{\Loc_{\calA}}^2\cdot
\size{\Edg_{\calA}})$. It remains to show alternating bisimilarity between
those structures. We define the relation
\[
R = \{(\ell,(\ell,\ell',m_{A_1},\ldots,m_{A_k}))\ \mid\ \ell\in\Loc_{\calA}, 
(\ell,\ell',m_{A_1},\ldots,m_{A_k}))\in\Loc_{\calB} \}.
\]
It is now only a matter of bravery to prove that $R$ is an alternating
bisimulation between~$\calA$ and~$\calB$.

\medskip
Point~\ref{iATS} is now immediate (through explicit~\CGS{}s), but it 
could also be proved in a similar way as point~\ref{eATS}.
\end{proof}

Let us mention that our translations are optimal (up to a polynomial): our
exponential translations cannot be achieved in polynomial time because of our
complexity results for \ATL model-checking. Note that it does not mean that
the resulting structures must have exponential size.

\begin{figure}[!ht]
\centering
\begin{minipage}{\linewidth}
\begin{minipage}{.45\linewidth}
\centering
\begin{tikzpicture}
  \tikzstyle{every node}=[circle,minimum size=6mm]
  \foreach \x / \y / \nod / \lab in 
    { 3.0 / 1.6 /b1/ b ,
      1.0 / 2.8 /s1/ a ,
      1.0 / 0.5 /s0/ d ,
      4.0 / 2.8 /a1/ c } 
  { \path (\x,\y)  node[draw] (\nod) {$\lab$};}
  \tikzstyle{every node}=[rectangle]
  \draw[arrows=-latex'] (s1.60) 
        .. controls +(60:6mm) and +(120:6mm) 
        .. (s1.120) node[midway,left=1pt] {$\scriptstyle\listof{3.1}$};
  \draw[arrows=-latex'] (s1) -- (a1) node[pos=.7,above=-1pt] {$\scriptstyle\listof{2.2+2.3}$};
  \foreach \noda / \nodb / \lab in 
    {s1/b1/1.1}
  {\draw[arrows=-latex'] (\noda) -- (\nodb) 
      node[midway,above=-1pt,sloped] {$\scriptstyle\listof{\lab}$};}
  \draw[arrows=-latex'] (s1) -- (s0) 
      node[midway,left=0pt] {$\scriptstyle\stackof{1.2+1.3+2.1}$}
      node[midway,right=0pt] {$\scriptstyle\stackof{3.2+3.3}$};
\end{tikzpicture}
\end{minipage}%
\begin{minipage}{.45\linewidth}

\centering
Moves from location~$A$:

\medskip
\begin{tabular}{|rl|}\hline
\multicolumn{2}{|c|}{Player $1$} \\\hline
move~$1$:& $\{b_{a,1,1},d_{a,1,2},d_{a,1,3}\}$  \\
move~$2$:& $\{c_{a,2,2},c_{a,2,3},d_{a,2,1}\}$  \\
move~$3$:& $\{a_{a,3,1},d_{a,3,2},d_{a,3,3}\}$  \\\hline
\end{tabular}

\medskip
\begin{tabular}{|rl|}\hline
\multicolumn{2}{|c|}{Player $2$} \\\hline
move~$1$:& $\{a_{a,3,1},b_{a,1,1},d_{a,2,1}\}$ \\
move~$2$:& $\{c_{a,2,2},d_{a,1,2},d_{a,3,2}\}$ \\
move~$3$:& $\{c_{a,2,3},d_{a,1,3},d_{a,3,3}\}$ \\\hline
\end{tabular}

\end{minipage}
\end{minipage}
\caption{Converting an~\CGSe into an~\ATS}\label{fig-cgs2ats} 
\end{figure}
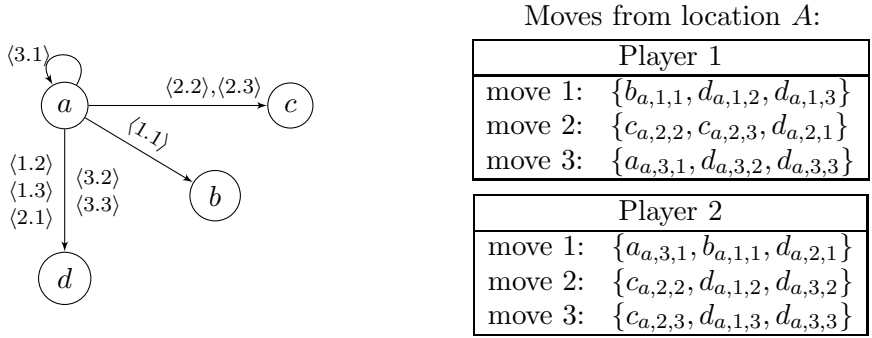

%
%

\subsection{Some remarks on the expressiveness of \texorpdfstring{\ATL}{ATL}}\label{expr}
\subsubsection{$\Diam[A]\Rel$ cannot be expressed with $\Diam[A]\Until$ and~$\Diam[A]\G$}

In the original papers defining~\ATL~\cite{focs1997-AHK,jacm49(5)-AHK}, the
syntax of that logic was slightly different from the one we used in this
paper: following classical definitions of the syntax of~\CTL, it was defined
as:
\begin{eqnarray*}
\ATLorig \ni \phi_s,\psi_s& ::=&\top\,\mid\, p \,\mid\, \non\phi_s \,\mid\, \phi_s\ou\psi_s
   \,\mid\, \Diam[A]\phi_p   \\
          \phi_p& ::= &  \X\phi_s\,\mid\, \G\phi_s \,\mid\, \phi_s\Until\psi_s.
\end{eqnarray*}

Duality is a fundamental concept in modal and temporal logics: for instance,
the dual of modality~\Until, often denoted by~\Rel and read \emph{release}, is
defined by $p \Rel q \equivdef \non((\non p)
\Until(\non q))$. Dual modalities allow, for instance, to put negations
inner inside the formula, which is often an important property when manipulating
formulas. 

In \LTL, modality~\Rel can be expressed using only~\Until and~\G: 
\begin{equation}
p \Rel q \equiv \G q \ou q \Until( p \et q).
\label{eqrel}
\end{equation}
In the same way, it is well known that \CTL can be defined using only
modalities \Ex\X, \Ex\G and \Ex\Until, and that we have
\begin{xalignat*}2
\Ex p \Rel q  &\equiv \Ex\G q \ou \Ex q \Until(p \et q) &
\All p \Rel q &\equiv \non\Ex(\non p)\Until(\non q).
\end{xalignat*}

It is easily seen that, in the case of \ATL, it is not the case that 
$\Diam[A] p \Rel q$ is equivalent to $\Diam[A]\G q \ou \Diam[A] q\Until (p \et
q)$: it could be the case that part of the outcomes satisfy $\G q$ and the
other ones satisfy $q\Until(p \et q)$. In fact, we prove that \ATLorig is
strictly less expressive than~\ATL:

\begin{theorem}
\label{th-expres}
There is no \ATLorig formula equivalent to $\Phi = \Diam[A](a\Rel b)$.
\end{theorem}

The proof of Theorem~\ref{th-expres} is based on techniques similar to
those used for proving expressiveness results for temporal logics like
\CTL or \ECTL~\cite{emerson90}: we build two families of
models~$(s_i)_{i\in\Nat}$ and~$(s'_i)_{i\in\Nat}$ s.t.~(1)~$s_i \not\sat
\Phi$, (2)~$s'_i\sat \Phi$ for any~$i$, and (3)~$s_i$ and $s'_i$
satisfy the same \ATLorig formula of size less than~$i$.
Theorem~\ref{th-expres} is a direct consequence of the existence of
such families of models. In order to simplify the presentation, the
theorem is proved for formula\footnote{This formula can also be
  written $\Diam[A]a\Wuntil b$, where $\Wuntil$ is the ``weak until''
  modality.} $\Phi=\Diam[A] ( b\Rel(a\ou b))$.

  The models are described by one single inductive
  \CGS\footnote{Given the translation from \CGS to~\ATS (see Section~\ref{traduc-models}), the result also
    holds for~\ATS{}s.}~$\mathcal C$, involving two players.  It is depicted
  on~Fig.~\ref{fig-models}.  
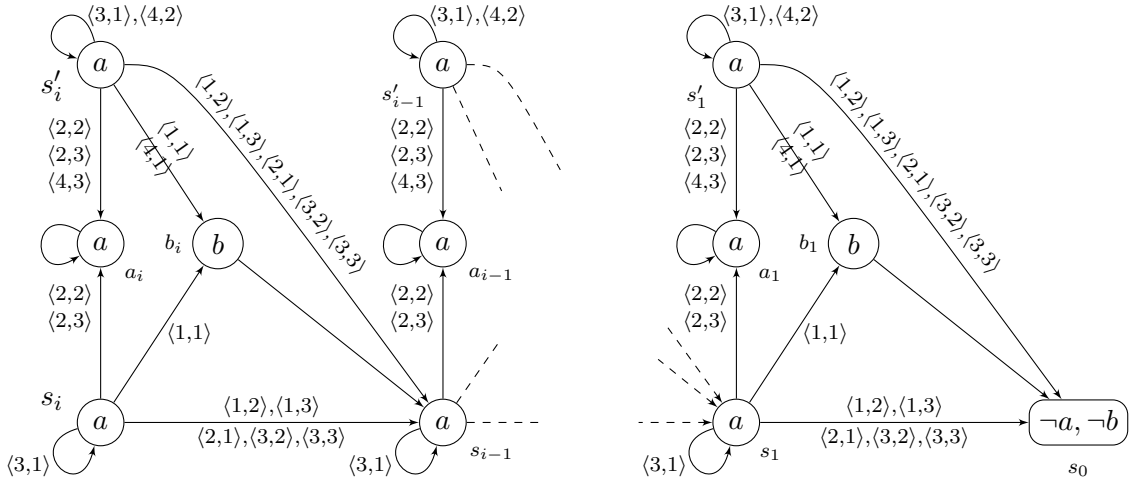
\begin{figure}[!ht]
\centering
\begin{tikzpicture}
\begin{scope}[xscale=1.3,yscale=1.4]
  \tikzstyle{every node}=[circle,minimum size=6mm]
  \foreach \x / \y / \nod / \lab / \nam in 
    { 1.0 / 3.0 /ai/ a / a_i,
      4.5 / 1.3 /sim1/ a /\quad s_{i-1},
      4.5 / 3.0 /aim1/ a /\quad a_{i-1},
      7.5 / 1.3 /s1/ a / s_1,
      7.5 / 3.0 /a1/ a / a_1} 
  { \path (\x,\y)  node[draw] (\nod) {$\lab$}
       -- +(.35,-.3) node {$\scriptstyle \nam$};}
  \foreach \x / \y / \nod / \lab / \nam in 
    { 2.2 / 3.0 /bi/ b /\ b_i, 
      8.7 / 3.0 /b1/ b /\ b_1}
  { \path (\x,\y)  node[draw] (\nod) {$\lab$}
       -- +(-.5,0) node {$\scriptstyle \nam$};}
  \foreach \x / \y / \nod / \lab / \nam / \hei in 
    { 1.0 / 1.3 /si/ a / s_i / .2,
      1.0 / 4.7 /spi/ a / s'_i / -.2}
  { \path (\x,\y)  node[draw] (\nod) {$\lab$}
       -- +(-.5,\hei) node {$\nam$};}
  \foreach \x / \y / \nod / \lab / \nam in 
    { 4.5 / 4.7 /spim1/ a / s'_{i-1},
      7.5 / 4.7 /sp1/ a / s'_1 }
  { \path (\x,\y)  node[draw] (\nod) {$\lab$}
       -- +(-.4,-.3) node {$\scriptstyle \nam$};}
  \path (11.0,1.3) node[draw,rectangle,rounded corners=2mm] (s0) {$\neg a,\neg b$}
       -- +(0,-.45) node {$\scriptstyle s_0$};
  \tikzstyle{every node}=[rectangle]
  \foreach \nod / \angle / \lab in 
    {si/-135/{3,1},sim1/-135/{3,1},s1/-135/{3,1},
     ai/-180/,aim1/180/,a1/180/}
    {\draw[arrows=-latex'] (\nod.\angle-30) 
        .. controls +(\angle-30:6mm) and +(\angle+30:6mm) 
        .. (\nod.\angle+30) node[midway,left=-1pt] {$\scriptstyle\listof{\lab}$};}
  \foreach \nod / \angle / \lab in 
    {spi/135/{3,1+4,2},spim1/135/{3,1+4,2},sp1/135/{3,1+4,2}}
    {\draw[arrows=-latex'] (\nod.\angle-30) 
        .. controls +(\angle-30:6mm) and +(\angle+30:6mm) 
        .. (\nod.\angle+30) node[pos=.3,right=1pt] {$\scriptstyle\listof{\lab}$};}
  \foreach \noda / \nodb / \lab in 
    {si/ai/{2,2+2,3},sim1/aim1/{2,2+2,3},s1/a1/{2,2+2,3}}
    {\draw[arrows=-latex'] (\noda) -- (\nodb) node[pos=.7,left=-1pt] {$\scriptstyle\stackof{\lab}$};}
  \foreach \noda / \nodb / \lab in 
    {spi/ai/{2,2+2,3+4,3},spim1/aim1/{2,2+2,3+4,3},sp1/a1/{2,2+2,3+4,3}}
    {\draw[arrows=-latex'] (\noda) -- (\nodb) node[pos=.5,left=-1pt] {$\scriptstyle\stackof{\lab}$};}
  \foreach \noda / \nodb / \lab in 
    {si/bi/{1,1},s1/b1/{1,1},bi/sim1/,b1/s0/}
    {\draw[arrows=-latex'] (\noda) -- (\nodb) node[midway,right=-1pt] {$\scriptstyle\listof{\lab}$};}
  \foreach \noda / \nodb / \lab in 
    {spi/bi/{1,1+4,1},sp1/b1/{1,1+4,1}}
    {\draw[arrows=-latex'] (\noda) -- (\nodb) node[midway,above=-12pt,sloped] {$\scriptstyle\stackof{\lab}$};}
  \foreach \noda / \nodb / \lab / \labb in 
    {si/sim1/{1,2+1,3}/{2,1+3,2+3,3},s1/s0/{1,2+1,3}/{2,1+3,2+3,3}}
    {\draw[arrows=-latex'] (\noda) -- (\nodb) 
      node[midway,above=-2pt] {$\scriptstyle\listof{\lab}$}
      node[midway,below=-2pt] {$\scriptstyle\listof{\labb}$};}
  \foreach \noda / \nodb / \lab in 
    {spi/sim1/{1,2+1,3+2,1+3,2+3,3},sp1/s0/{1,2+1,3+2,1+3,2+3,3}}
    {\draw[arrows=-latex'] (\noda) .. controls +(0:8mm) .. (\nodb) 
      node[pos=.71,above=-1pt,sloped] {$\scriptstyle\listof{\lab}$};}

  \draw[style=dashed] (sim1) -- +(10mm,0mm);
  \draw[style=dashed] (sim1) -- +(6mm,8mm);
    \draw[style=dashed] (spim1) -- +(6mm,-12mm);
    \draw[style=dashed] (spim1) .. controls +(6mm,0mm) .. +(12mm,-10mm);
    \draw[style=dashed,arrows=latex'-] (s1) -- +(left:1cm);
    \draw[style=dashed,arrows=latex'-] (s1) -- +(-8mm,6mm);
    \draw[style=dashed,arrows=latex'-] (s1) -- +(-7mm,9mm);
\end{scope}
\end{tikzpicture}
\caption{The \CGS~$\mathcal C$, with states~$s_i$ and~$s'_i$ on the left}
\label{fig-models}
\end{figure} 
  A~label
  $\listof{\alpha,\beta}$ on a transition indicates that this
  transition corresponds to move~$\alpha$ of player~$A_1$ and to
  move~$\beta$ of player~$A_2$.  In that~\CGS, states~$s_i$ and~$s'_i$
  only differ in that player~$A_1$ has a fourth possible move
  in~$s'_i$.  This ensures that, from state~$s'_i$ (for any~$i$),
  player~$A_1$ has a strategy (namely, he should always play~$4$) for
  enforcing~$a\Wuntil b$. But this is not the case from state~$s_i$:
  by induction on~$i$, one can prove $s_i \not\sat \Diam[A_1] a \Wuntil
  b$. The base case is trivial. Now assume the property holds for $i$:
  from $s_{i+1}$, any strategy for $A_1$ starts with a move in
  $\{1,2,3\}$ and for any of these choices, player $A_2$ can choose a
  move (2, 1 and 2 resp.) that enforce the next state to be $s_i$
  where by i.h. $A_1$ has no strategy for~$a\Wuntil b$.

  We now prove that $s_i$ and $s'_i$ satisfy the same ``small''
  formulae.  First, we have the following equivalences:
\newcounter{lemmappun}
\setcounter{lemmappun}{\value{lemma}}
\begin{lemma}\label{lemma1}
For any $i>0$, for any $\psi\in\ATLorig$ with~$\size\psi\leq i$:
 \begin{xalignat}1
 b_i\models \psi &\text{ iff } b_{i+1}\models \psi \label{eq1}\\
 s_i\models \psi &\text{ iff } s_{i+1}\models \psi \label{eq2}\\
 s'_i\models \psi &\text{ iff } s'_{i+1}\models \psi \label{eq3}
 \end{xalignat}
\end{lemma}

\begin{proof}
The proof proceeds by induction on~$i$, and on the structure 
of the formula~$\psi$.

\paragraph{Base case: $i=1$.} Since we require that $\size\psi\leq i$, $\psi$
can only be an atomic proposition. The result is then obvious.

\paragraph{Induction step.} We assume the result holds up 
to some~$i-1\geq 1$, and prove that it then still holds for~$i$. 
Let $\psi$ s.t. $\size\psi\leq i$. We now proceed by structural induction on~$\psi$:

\begin{enumerate}[$\bullet$]
\item The result is again obvious for atomic propositions, as well as for
  boolean combinations of subformulae.
  
\item Otherwise, the ``root'' combinator of~$\psi$ is a modality. If it is a
  \CTL modality, the results are quite straightforward. Also, since there is
  only one transition from~$b_i$, any \ATLorig modality can be expressed as a \CTL
  modality in that state, and~\eqref{eq1} follows.

\item If $\psi = \Diam[A_1]\X \psi_1$: 
  Assume~$s_{i}\models\psi$. Then, depending on the strategy, either $b_i$
  and~$s_{i-1}$, or $a_i$ and~$s_{i-1}$, or $s_i$ and~$s_{i-1}$, should
  satisfy~$\psi_1$. By i.h., this propagates to the next level, and the same strategy
  can be mimicked from~$s_{i+1}$.  

  The converse is
  similar (hence~\eqref{eq2}), as well as the proof for~\eqref{eq3}.

\item If $\psi = \Diam[A_1]\G \psi_1$: 
  If~$s_{i}\models\psi$, then $s_{i}$, thus~$s_{i+1}$, satisfy~$\psi_1$.
  Playing move~$3$ is a strategy for player~$A_1$ to enforce~$\G\psi_1$
  from~$s_{i+1}$, since the game will either stay in~$s_{i+1}$ or go to~$s_i$,
  where player~$A$ has a winning strategy. 

  The converse is immediate, as player~$A_1$ cannot avoid~$s_{i}$ when
  playing from~$s_{i+1}$. Hence~\eqref{eq2} for~$\Diam[A_1]\G$-formulae.

\medskip

  If~$s'_{i}\models\psi$, then both~$s'_{i}$ and~$s'_{i+1}$
  satisfy~$\psi_1$. Also, player~$A_1$ cannot avoid the play to go in
  location~$s_{i-1}$. Thus, $s_{i-1}\models\psi_1$ ---and by~i.h., so
  does~$s_{i}$--- and $s_i\models\psi$, as above. 
  Now, following the same strategy in~$s'_{i+1}$ as the winning strategy
  of~$s'_i$ clearly enforces~$\G\psi_1$. 
  The converse is similar: it suffices to mimic, from~$s'_{i}$, 
  the strategy witnessing the fact that~$s'_{i+1}\models\psi$. This
  proves~\eqref{eq3}, and concludes this case.

\item If $\psi = \Diam[A_1]\psi_1\Until\psi_2$: 
  If $s_{i}\models\psi$, then either~$\psi_2$ or~$\psi_1$ holds in~$s_i$, thus
  in~$s_{i+1}$. The former case is trivial. In the latter, player~$A_1$ can
  mimic the winning strategy in~$s_{i+1}$: the game will end up in~$s_i$, with
  intermediary states satisfying~$\psi_1$ (or~$\psi_2$), and he can then apply
  the original strategy.

  The converse is obvious, since from~$s_{i+1}$, player~$A_1$ cannot avoid
  location~$s_{i}$, from which he must also have a winning strategy.\medskip

  If $s'_{i}\models\psi$, omitting the trivial case where $s'_{i}$
  satisfies~$\psi_2$, we have that~$s_{i-1}\models\psi$.
  Also, a (state-based) strategy in~$s'_i$ witnessing~$\psi$ necessary
  consists in playing move~$1$ or~$2$. Thus $a_i$ and~$b_i$ satisfy~$\psi$,
  and the same strategy (move~$1$ or~$2$, resp.) enforces~$\G\psi_1$
  from~$s_i$. It is now easy to see that the same strategy is correct
  from~$s'_{i+1}$. 
  Conversely, apart from trivial cases, the strategy can again only consist in
  playing moves~$1$ or~$2$. In both cases, the game could end up in~$s_{i}$,
  and then in~$s_{i-1}$. Thus $s_{i-1}\models\psi$, and the same strategy as
  in~$s'_{i+1}$ can be applied in~$s'_{i}$ to witness~$\psi$.

\item The proofs for $\Diam[A_2]\X \psi_1$, $\Diam[A_2]\G \psi_1$,
  and~$\Diam[A_2]\psi_1\Until\psi_2$ are very similar to the previous ones. \qed
\end{enumerate}
\end{proof}

%
%

\begin{lemma}
$\forall i>0$, $\forall \psi\in\ATLorig$ with~$\size\psi\leq i$:
$s_i\models \psi \;\text{ iff }\; s'_{i}\models \psi$.
\end{lemma}

\begin{proof}
  The proof proceeds by induction on~$i$, and on the structure of the
  formula~$\psi$. The case $i=1$ is trivial, since $s_1$ and $s'_1$
  carry the same atomic propositions. For the induction step,
  dealing with \CTL modalities (\Diam[\emptyset] and
  \Diam[A_1,A_2]) is also straightforward, then we just consider
  \Diam[A_1]{}- and~\Diam[A_2]{}-modalities. 

  First we  consider \Diam[A_1]{}-modalities. 
  It is well-known that we can restrict to state-based strategies in this
  setting. 
  If player~$A_1$ has a
  strategy in~$s_i$ to enforce something, then he can follow the same
  strategy from~$s'_i$. Conversely, if player~$A_1$ has a strategy in~$s'_i$ to
  enforce some property, two cases may arise: either the strategy consists in
  playing move~$1$, $2$ or~$3$, and it can be mimicked from~$s_i$. Or the
  strategy consists in playing move~$4$ and we distinguish three cases:
  \begin{enumerate}[$\bullet$]
  \item $\psi=\Diam[A_1]\X\psi_1$: that move~$4$ is a winning strategy entails  
  that~$s'_i$, $a_i$ and~$b_i$ must satisfy~$\psi_1$.
  Then $s_i$ (by~i.h. on the formula) and~$s_{i-1}$ (by Lemma~\ref{lemma1}) both
  satisfy~$\psi_1$. Playing move~$1$ (or~$3$) in $s_i$ ensures that the next
  state will satisfy~$\psi_1$.
\item $\psi=\Diam[A_1]\G\psi_1$: by playing move~$4$, 
  the game could end up in~$s_{i-1}$ (\emph{via}~$b_i$), and in~$a_i$
  and~$s'_i$. Thus $s_{i-1}\models\psi$, and in particular~$\psi_1$. By~i.h.,
  $s_i\models\psi_1$, and playing move~$1$ (or~$3$) in~$s_i$, and then
  mimicking the original strategy (from~$s'_i$), enforces~$\G\psi_1$.
\item $\psi=\Diam[A_1]\psi_1\Until\psi_2$: a strategy starting with
  move~$4$ implies $s'_i\models\psi_2$ (the game could stay in~$s'_i$ for
  ever). Then $s_i\models\psi_2$ by~i.h., and the result follows.
  \end{enumerate}

  \noindent We now turn to $\Diam[A_2]$-modalities: clearly
  if~$\Diam[A_2]\psi_1$ holds in~$s'_i$, it also holds in~$s_i$.
  Conversely, if player~$A_2$ has a (state-based) strategy to enforce
  some property in~$s_i$: If it consists in playing moves~$1$ or~$3$,
  then the same strategy also works in~$s'_i$. Now if the strategy
  starts with move~$2$, then playing move~$3$ in~$s'_i$ has the same
  effect, and thus enforces the same property. \popQED
\end{proof}

\begin{remark}
\ATLorig and \ATL have the same distinguishing power as the
fragment of~\ATL involving only the $\Diam\X$ modality (see~\cite[proof of
  Th.~6]{AHKV-concur98}). This means 
that we cannot exhibit two models $M$ and~$M'$ s.t.~(1)~$M \sat
\Phi$, (2)~$M'\not\sat \Phi$, and (3)~$M$~and~$M'$ satisfy the same
\ATLorig formula.
\end{remark}

\begin{remark}
In~\cite{jacm49(5)-AHK}, a
restriction of \CGS{}~---the turn-based 
\CGS{}s--- is considered. In any location of these models~(named TB-\CGS
hereafter),  only one player has several moves (the other players have
only one possible choice). Such models have the property of \emph{determinedness}: given a
set of players~$A$, either there is a strategy for $A$ to win some
objective~$\Phi$, or there is a strategy for other players ($\Agt
\backslash A$) to enforce~$\non\Phi$. In such systems, 
modality~\Rel can be expressed as follows:
\(
\Diam[A] \phi \Rel \psi \equiv_{\mbox{\tiny TB-\CGS}} 
\non\Diam[\Agt\backslash A] (\non \phi) \Until (\non \psi). 
\)
\end{remark}

\subsubsection{$\Diam[A] \Ginf$ and $\Diam[A] \Finf$ cannot be expressed in \ATL}
\label{sec-expre-finf}
It is well known that \ECTL formulae of the form $\Ex\Finf P$ (and its
dual $\All\Ginf P$) cannot be expressed in \CTL~\cite{emerson90}. 
On the other hand, the following equivalences hold:
\begin{xalignat*}2
\Ex\Ginf P &\equiv \Ex\F\Ex\G P &
\All\Finf P &\equiv \All\G\All\F P.
\end{xalignat*}

The situation is again different in~\ATL: neither~$\Diam[A]\Finf$ nor
$\Diam[A]\Ginf$ are expressible in~\ATL.
Indeed, assume that $\Diam[A]\Finf$ could be expressed by the \ATL
formula~$\Phi$. This holds in particular in $1$-player games (\ie, Kripke
structures). In the case where coalition~$A$ contains the only player, we
would end up with a \CTL equivalent of~$\Ex\Finf$, which is known not to
exist.
A~similar argument applies for~$\Diam[A]\Ginf$.

\section{Conclusion}

In this paper, we considered the basic questions of expressiveness and
complexity of~\ATL.  
We precisely characterized the complexity of \ATL, \ATL+, \EATL
and \ATL* model-checking, on both ATSs and~CGSs, when the number of
agents is not fixed.
These results complete the previously known results  about these
formalisms (and corrects some of them). It is interesting to see that their complexity classes
(\DD2~or~\DD3) are unusual in the area of model-checking.
We also showed that \ATL, as originaly defined
in~\cite{focs1997-AHK,compos1997-AHK,jacm49(5)-AHK}, is not as expressive as it could
be expected, and we argue that the modality ``Release'' should be
added in its definition.

\bibliographystyle{alpha}
\bibliography{ATL}
\vskip-50 pt
\end{document}